\documentclass[12pt]{article}

\usepackage[a4paper, left=2.5cm, right=2.5cm, top=2.5cm, bottom=2.5cm]{geometry}
\usepackage{amssymb}
\usepackage{amsmath}
\usepackage{amsthm}
\usepackage{array}
\usepackage{booktabs}
\usepackage{multirow}
\usepackage{graphicx}
\usepackage{caption}
\usepackage{subcaption}
\usepackage{algorithm}
\usepackage{algpseudocode}
\usepackage{xcolor}
\usepackage{comment}
\usepackage{makecell}
\usepackage{float}
\usepackage{placeins}
\usepackage{tikz}
\usetikzlibrary{positioning,arrows.meta,shadows,fit}
\usepackage{csquotes}
\usepackage{enumitem}
\usepackage{xurl}
\usepackage[authoryear]{natbib}

\theoremstyle{plain}
\newtheorem{theorem}{Theorem}[section]
\newtheorem{lemma}{Lemma}[section]

\theoremstyle{definition}

\theoremstyle{remark}
\newtheorem{remark}{Remark}[section]

\DeclareMathOperator*{\argmin}{arg\,min}

\usepackage[bookmarks=false]{hyperref}

\hypersetup{
  pdftitle  = {QUBO-Based Calibration for Regression Trees},
  pdfauthor = {Iro René Kouarfate, Maxime Dion, Anne MacKay, Mathieu Pigeon}
}

\title{
QUBO-Based Calibration for Regression Trees\thanks{
Submitted to the editors DATE.

\medskip
\noindent
Funding: The authors acknowledge financial support from the Natural Sciences and Engineering Research Council of Canada [grant number 2024-05794; 2025-04224], the Pôle de recherche en finance quantique,
the Ministère de l’Économie, de l’Innovation et de l’Énergie du Québec, the Canada First Research Excellence Fund,
and the Chaire Co-operators en analyse des risques actuariels.
}
}

\author{
Iro René Kouarfate\thanks{
Department of Mathematics, Université de Sherbrooke,
Sherbrooke, Quebec, Canada
}\;;
Maxime Dion\thanks{
Institut quantique, Université de Sherbrooke,
Sherbrooke, Quebec, Canada
}\;; \\
Anne MacKay\footnotemark[2]\;;
and Mathieu Pigeon\thanks{
Department of Mathematics, Université du Québec à Montréal,
Montreal, Quebec, Canada
}
}

\date{}

\begin{document}

\maketitle

\begin{abstract}

Tree-based regression models are widely used in supervised learning, with the Classification and Regression Tree (CART) algorithm serving as a standard reference. CART construction involves solving a sequence of split-selection optimization problems, which are fractional and become combinatorial in nature for categorical predictors. Although, in the single-target regression setting with squared-error loss, this problem admits an efficient exact solution, as shown in \cite{fisher1958grouping, breiman1984cart}, we adopt a QUBO framework to address the categorical split-selection problem. This choice is motivated by the general-purpose nature of QUBO formulations, which provide a unified optimization framework that naturally extends to settings where classical splitting strategies rely on heuristics and may yield suboptimal solutions.

We propose a QUBO formulation of the categorical split-selection problem in single-target least-squares regression. The fractional nature of the objective function is handled using Dinkelbach’s algorithm \citep{dinkelbach1967algorithm} together with a class-based encoding. This leads to a compact sequence of QUBO problems whose size depends only on the number of categories, rather than on the sample size. Using state-of-the-art QUBO solvers, we construct QUBO-based regression trees with predictive performance comparable to standard CART, while yielding higher-quality categorical splits. Overall, this work highlights the relevance of QUBO formulations as a flexible optimization framework for tree-based learning and opens perspectives for future hybrid classical–quantum optimization approaches within CART extensions.

\bigskip

    
\end{abstract}



\maketitle
\section{Introduction}
Tree-based methods are a central class of nonparametric statistical learning techniques widely used for regression and classification due to their interpretability, flexibility, and ability to handle both continuous and categorical predictors. Decision trees rely on recursive partitioning of the feature space, where splits are selected to maximize homogeneity within nodes according to criteria such as variance reduction or impurity measures \citep{breiman1984cart}. Among classical decision tree algorithms, including CHAID \citep{kass1980exploratory}, ID3 \citep{quinlan1986induction}, and C4.5 \citep{quinlan1993c45}, the CART framework remains a widely adopted reference, providing a systematic methodology for tree construction, pruning, and complexity control. 

The CART growing procedure relies on solving a sequence of split-selection optimization problems on a single explanatory variable at each node. While this task is tractable for continuous predictors, categorical variables induce a combinatorial optimization problem whose complexity grows exponentially with the number of categories. In the single-target regression setting with squared-error loss, however, the binary split-selection problem for a categorical predictor admits an efficient exact method. As shown in \citep{fisher1958grouping, breiman1984cart}, the optimal split can be obtained by ordering the categories according to their empirical mean response and evaluating the resulting contiguous partitions. This so-called “sort and scan” procedure reduces the computational complexity from exponential to linear in the number of categories while preserving optimality within this framework. Despite its efficiency, this method does not extend to the multivariate setting (\citep{segal1992tree, deathMultivariate2002}) $\mathbf{Y} \in \mathbb{R}^d$ $(d \geq 2)$, where no total order exists on $\mathbb{R}^d$ (\cite{larsen2004multivariate}). Heuristic strategies, such as sorting categories by a scalar projection of the within-category mean vector $\bar{\mathbf{Y}}_k$, are typically used but provide no optimality guarantee (\citep{coppersmith1999partitioning, wright2019splitting}).

In this work, we address this limitation by reformulating the categorical split-selection problem as a Quadratic Unconstrained Binary Optimization (QUBO) problem, enabling direct optimization over the full space of binary partitions in the single-target least-squares regression. While the present work does not investigate global optimization in more general settings, the proposed formulation naturally provides a general-purpose combinatorial optimization framework applicable to such extensions, including multi-target regression, complex partitioning schemes, and "look-ahead" strategies \citep{breiman1984cart}. As a flexible generalization of the classical CART splitting mechanism, the QUBO formulation also opens the perspectives to integrating hybrid classical–quantum optimization approaches within the CART framework and its extensions.

QUBO has emerged as a flexible modeling framework for combinatorial optimization problems involving discrete decision variables. In QUBO formulations, the objective function is expressed as a quadratic polynomial in binary variables, enabling the use of both classical optimization solvers and quantum-inspired hardware to compute optimal or near-optimal solutions \citep{glover2022quantum}. Interest in QUBO has grown significantly in recent years, driven by the development of efficient optimization solvers and its equivalence with the Ising model from statistical physics, which plays a central role in quantum computing \citep{lucas2014ising}. In particular, many quantum architectures, such as quantum annealers, are designed to minimize Ising Hamiltonians. By exploiting qubits and intrinsic quantum phenomena, quantum computing extends beyond classical bit-based computation. Therefore, expressing combinatorial optimization problems in QUBO form provides a natural and effective approach for exploiting quantum hardware. Such formulations are particularly well suited to partitioning and assignment problems. QUBO models have been successfully applied to various domains, including clustering \citep{matsumoto2022distance}, facility location \citep{aoki2021quboKMedoids}, Gaussian process variance reduction \citep{otsuka2020quboGaussian}, and financial portfolio optimization \citep{phillipson2021portfolio}. 

The split-selection criterion in regression trees arises from variance reduction and has a ratio-type structure, which places it within the class of fractional programming problems \citep{dinkelbach1967algorithm, schaible1976fractional}. Dinkelbach’s algorithm provides a well-established approach for handling such objectives through a sequence of parametric subproblems and has been successfully applied in various optimization settings \citep{you2009dinkelbach, zhong2014globally, ajagekar2020quantum}. 

In the context of decision trees, \citet{yawata2022qubo} propose a QUBO-based framework that enables multidimensional decision boundaries through logical-product split conditions, illustrating the potential of QUBO formulations for tree-based learning. Their approach unifies continuous and categorical predictors via feature binarization and discretizations, but modifies the classical variance-reduction objective by removing its fractional component, thereby altering its original statistical interpretation. Moreover, categorical variables are treated as collections of binary indicators, rather than preserving their natural partitioning structure. Consequently, QUBO formulations that preserve the exact fractional structure of the CART splitting criterion—naturally expressed as a ratio of quadratic forms—remain relatively unexplored. While multidimensional extension is interesting, it is not addressed in the present work. Instead, we focus on an exact QUBO reformulation of the categorical split-selection problem in regression CART that preserves the original variance criterion. 

For decision tree split selection, QUBO models require binary encoding for categorical variables. Observation-based encoding is commonly used, but it is not well suited to the node-wise CART splitting process and scales with the sample size, leading to large QUBO instances and high computational cost \citep{matsumoto2022distance}. To address these limitations, we propose a class‑based encoding strategy in which binary variables represent category assignments at the split level rather than at the observation level. 

We introduce a QUBO-based regression tree construction in which the fractional nature of the split-selection objective is handled using Dinkelbach’s algorithm together with a class-based encoding, resulting in a sequence of compact QUBO problems that can be solved using classical optimization solvers. The main contributions of this work are summarized below.
\begin{enumerate}[label=\Roman*.]
\setlength{\itemsep}{-0.25em}
\item \textbf{A compact class-based QUBO encoding for categorical variables}. We introduce an encoding strategy whose size scales with the number of categories rather than the number of observations, substantially reducing the dimensionality of the resulting QUBO compared to observation-based approaches.
\item \textbf{An exact QUBO formulation of the CART variance-reduction split-selection criterion for categorical predictors}. Our formulation provides an exact treatment of the ratio-type split-selection objective underlying CART and enables the computation of optimal or near-optimal solutions to the resulting combinatorial optimization problem using QUBO solvers.
\item \textbf{A theoretically grounded initialization scheme for the optimization parameter}. We propose an initialization strategy with theoretical justification that ensures practical convergence of the iterative algorithm based on Dinkelbach’s method.
\item \textbf{A QUBO-based regression tree implementation for prediction}.
We present a classical implementation of regression tree based on a QUBO formulation of the objective function for selecting splits of categorical variables, paving the way for hybrid classical-quantum CART algorithms.
\end{enumerate}

The remainder of the paper is organized as follows. Section~2 presents the proposed QUBO-based split formulation and theoretical results. Numerical experiments are reported in Section~3. Section~4 concludes and outlines directions for future research.


\section{QUBO-Based Regression Tree}

This section presents the theoretical foundations of our approach. We review regression trees and the QUBO framework, then reformulate the regression tree node-splitting cost as a pairwise interaction function based on response values. Since the empirical risk in~\eqref{eq:reg_tree_split} is not directly compatible with standard QUBO encodings, we propose an iterative approach with a dedicated encoding for categorical variables, leading to an explicit QUBO formulation for regression tree splits. The main implementation steps of the proposed algorithm are also outlined. Proofs are provided in Appendix~\ref{app:proofLemma} in the Online Supplement.

Throughout this section, we consider a continuous random response variable $Y$, a $d$-dimensional feature vector $\mathbf{X} = (X_1,\ldots,X_d)$ defined on a feature space $\mathcal{X}$, and an i.i.d. training sample $\mathcal{L} = \{(\mathbf{X}_i, Y_i)\}_{i=1}^N$. Let $\mathcal{P}$ denote a functional space of predictors $\pi : \mathcal{X} \to \mathbb{R}$. We denote by $\mathcal{N}_t$ a given parent node in a regression tree, with left and right child nodes $\mathcal{N}_{t_L}$ and $\mathcal{N}_{t_R}$, where $\mathcal{T} \supset \{t, t_L, t_R\}$ denotes the set of node indices. $\mathcal{N}_t$ contains $N_S$ observations, while $\mathcal{N}_{t_L}$ and $\mathcal{N}_{t_R}$ contain $N_L$ and $N_R$ observations, respectively, with $N_S = N_L + N_R$. Each node is associated with an empirical response variance, defined as the sample variance of the responses within the node and denoted by $\operatorname{Var}_S$, $\operatorname{Var}_L$, and $\operatorname{Var}_R$.

\subsection{Tree-based Regression}

In regression analysis, the objective is to approximate the relationship between $Y$ and $\mathbf{X}$ by a predictor 
$\pi \in \mathcal{P}$. Given a loss function $\phi$, the associated risk is defined as $\mathbb{E}[\phi(Y,\pi(\mathbf{X}))]$, and the optimal predictor satisfies
\begin{equation}
\pi_0 = \arg\min_{\pi \in \mathcal{P}} \mathbb{E}[\phi(Y,\pi(\mathbf{X}))].
\label{eq:pi0}
\end{equation}
When the squared error loss is used, the solution of \eqref{eq:pi0} is the regression function $\pi_0(\mathbf{X}) = \mathbb{E}(Y \mid \mathbf{X})$. 

Tree-based regression models are obtained by restricting the functional space $\mathcal{P}$ to piecewise-constant functions defined on a recursive partition of the feature space \citep{breiman1984cart}. Each region of the partition is associated with a constant prediction equal to the empirical mean of the observations it contains. Since the joint distribution of $(\mathbf{X},Y)$ is unknown, both the partition and the node predictions are learned from the training data through empirical risk minimization. At each internal node $\mathcal{N}_t$, tree construction requires selecting a binary split that locally minimizes the empirical squared error after partitioning the node into two child nodes $\mathcal{N}_{t_L}$ and $\mathcal{N}_{t_R}$. This leads to the node-splitting optimization problem
\begin{equation}
\argmin_{\mathcal{N}_{t_L},\, \mathcal{N}_{t_R}}
\left\{
\sum_{\mathbf{X}_i \in \mathcal{N}_{t_L}} (Y_i - \bar{Y}_{t_L})^2
+
\sum_{\mathbf{X}_i \in \mathcal{N}_{t_R}} (Y_i - \bar{Y}_{t_R})^2
\right\},
\label{eq:reg_tree_split}
\end{equation}
where $\bar{Y}_{t_L}$ and $\bar{Y}_{t_R}$ are the empirical means within the two child nodes. This criterion can equivalently be expressed as a weighted sum of the empirical variances. When predictor variables are categorical, the split-selection problem in \eqref{eq:reg_tree_split} becomes a discrete combinatorial optimization problem, since it amounts to assigning categories to either the left or the right child node. This combinatorial structure motivates the use of a Quadratic Unconstrained Binary Optimization framework.

\subsection{QUBO Framework}

QUBO provides a general framework for formulating combinatorial optimization problems defined over binary decision variables \citep{glover2022quantum,kochenberger2014ubiquitous}. In a QUBO model, both linear effects and pairwise interactions between binary variables are encoded in a quadratic objective function, allowing a wide range of discrete decision problems to be represented in a unified form. A key motivation for using QUBO lies in its dual compatibility with classical and quantum optimization paradigms. On the classical side, the matrix-based structure of QUBO formulations enables efficient implementations using mature optimization solvers such as \texttt{Gurobi} or \texttt{CPLEX}. On the quantum side, QUBO problems are mathematically 
equivalent to Ising models from statistical physics, which form the computational basis of quantum annealers and hybrid quantum–classical optimization platforms \citep{lucas2014ising,kadowaki1998quantum}.
From a modeling perspective, QUBO is particularly well suited to problems involving binary assignment decisions, such as the allocation of categorical levels to child nodes in regression tree splitting. Moreover, constrained binary optimization problems can be naturally reformulated within the QUBO framework through the use of quadratic 
penalty terms, enabling flexible modeling while preserving an unconstrained objective function. These properties are exploited to encode the categorical split-selection problem arising in 
tree-based regression as a QUBO model. The formal definition of QUBO problems, together with standard techniques for handling constraints via penalty methods, are recalled in Appendix~\ref{Appendix: QUBO} in the Online Supplement for completeness.

\subsection{Reformulation of the Node Splitting Cost Function}

The proposed approach integrates the fractional programming framework of \citet{dinkelbach1967algorithm} with a mean-free variance identity, which allows the variance to be expressed directly from the observations without involving the sample mean. This leads to a reformulation of the cost function solely in terms of pairwise differences.

\begin{lemma}[Variance Identity]
\label{lemma:variance_identity}
Let $\{Y_1, Y_2, \ldots, Y_N\}$ be a finite sample with mean $\bar{Y}$ and variance $\sigma^2$. Then, the empirical variance can be equivalently written as
\begin{equation}
\label{eq:variance_identity}
    \sigma^2 
    = \frac{1}{N} \sum_{i=1}^N (Y_i - \bar{Y})^2
    = \frac{1}{2N^2} 
    \sum_{i=1}^{N} \sum_{j=1}^{N} (Y_i - Y_j)^2.
\end{equation}
\end{lemma}

This lemma states a classical identity for the empirical variance. We include it for completeness, and a proof is provided in Appendix~\ref{app:proofLemma} in the Online Supplement.

By Lemma~\ref{lemma:variance_identity}, the original node-splitting 
optimization problem in \eqref{eq:reg_tree_split} can be reformulated as a fractional programming problem, that is


\begin{align}
\label{eq:proposed_fractional}
    \argmin _{\mathcal{N}_{t_L},\, \mathcal{N}_{t_R}} 
    \left\{ 
    \frac{\sum\limits_{i,j \in \mathcal{N}_{t_L}} (Y_i - Y_j)^2}{2N_L}
    + 
    \frac{\sum\limits_{i,j \in \mathcal{N}_{t_R}} (Y_i - Y_j)^2}{2N_R}
    \right\}.
\end{align} 


\subsection{Iterative Approach}

To handle the fractional structure of problem \eqref{eq:proposed_fractional}, 
we employ Dinkelbach's algorithm \citep{dinkelbach1967algorithm}, which transforms a fractional programming problem into an equivalent parametric formulation. The following definition and lemma, established in \citet{dinkelbach1967algorithm} and extended to mixed-integer settings in \citet{zhong2014globally}, provide the theoretical 
foundation of this approach.

\begin{lemma}[Parametric reformulation
\citep{dinkelbach1967algorithm}]
\label{lemma:dinkelbach}
Let $S$ be a nonempty compact feasible set, and $\mathbf{n}, \mathbf{d} : S \to \mathbb{R}$ be continuous with $\mathbf{d}(x) > 0$ for all $x \in S$. The fractional program
\begin{equation*}
    (P): \quad  \min_{x \in S}\; \{\mathbf{n}(x)/ \mathbf{d}(x)\},
\end{equation*}
admits for $\lambda \geq 0$ the equivalent parametric formulation
\begin{equation*}
    (P_\lambda): \quad F(\lambda) := \min_{x \in S}  F(\lambda, x),\; \text{with}\;\;  F(\lambda, x) = 
    \mathbf{n}(x) - \lambda\, \mathbf{d}(x).
\end{equation*}
\end{lemma}

Based on Lemma~\ref{lemma:dinkelbach}, Lemma~\ref{lemma:optim_condition} 
characterizes the optimal solutions of $(P)$ and its parametric equivalent $(P_\lambda)$.

\begin{lemma}[Optimality Condition \citep{dinkelbach1967algorithm}]
\label{lemma:optim_condition}
Let $x^* \in S$ be a solution to $(P)$ with optimum value
$\lambda^* = \mathbf{n}(x^*)/\mathbf{d}(x^*)$.
Then,
\begin{equation*}
    \lambda^* = \min_{x \in S}\;\frac{\mathbf{n}(x)}{\mathbf{d}(x)}
    \iff
    F(\lambda^*) := \min_{x \in S}\, F(\lambda^*, x) = 0,
\end{equation*}
in which case $x^*$ also solves $(P_{\lambda^*})$, i.e.,
$x^* \in \arg\min_{x \in S}\, F(\lambda^*, x)$.
Furthermore, starting from any $\lambda_0 \geq 0$, the Dinkelbach 
update rule $\lambda_{k+1} = \frac{\mathbf{n}(x^*_k)}{\mathbf{d}(x^*_k)}$, where $x^*_k \in \arg\min_{x \in S}\; F(\lambda_k, x)$,
converges superlinearly to the optimal pair $(\lambda^*, x^*)$.
\end{lemma}

Proofs are provided in \cite{dinkelbach1967algorithm} and \cite{schaible1976fractional}.

In other words, optimality is reached when $\lambda_k$ matches the optimal ratio $\lambda^* = \mathbf{n}(x^*)/\mathbf{d}(x^*)$, driving the residual $F(\lambda^*) = F(\lambda^*, x^*)$ to zero. Next, we show that each subproblem $(P_{\lambda_k})$ can be reformulated as a QUBO problem.

\subsection{Binarization of Categorical Explanatory Variable}

QUBO models are defined over binary decision variables. To express the node-splitting optimization problem as a QUBO model, we first define a binarization of categorical explanatory variables using \textit{class-based encoding} strategy. Let $X_j$ be a categorical predictor with $M$ distinct categories $C_\alpha$, $\alpha = 1, \dots, M$. Denoting by $\mathbf{1}_{A}$ the indicator function of event $A$, for each category or class, we introduce a binary variable $q_\alpha$ such that
\begin{equation}
\label{eq:q_alphabeta}
    q_\alpha = \mathbf{1}_{\{C_\alpha \in \mathcal{N}_{t_L}\}}=
    1 - \mathbf{1}_{\{C_\alpha \in \mathcal{N}_{t_R}\}}.
\end{equation}

A categorical split is therefore fully characterized by the binary vector $\mathbf{q} = (q_{1},\dots,q_{M})^\top$. This representation encodes $2^M - 2$ admissible partitions of the category set, excluding the trivial vectors $\mathbf{q} = \mathbf{0}$ and $\mathbf{q} = \mathbf{1}$. These vectors assign all categories to a single child node, leaving the other empty, and therefore correspond to the absence of any effective split, that is, the trivial split
$\{\mathcal{N}_t,\emptyset\}$ rather than a valid partition. This binary encoding can be directly embedded into an objective function.

\subsection{Explicit QUBO Formulation of Regression Tree Splits}

The QUBO model is defined as
\begin{equation*}
    \min_{\mathbf{q} \in \mathcal{Q}^*} \mathbf{q}^\top Q \mathbf{q},
\end{equation*}
where $\mathcal{Q}^* = \{0,1\}^M \setminus \{\mathbf{0}, \mathbf{1}\}$ 
denotes the set of admissible binary vectors, i.e., all vectors excluding the trivial non-split solutions $\mathbf{q} = \mathbf{0}$ and $\mathbf{q} = \mathbf{1}$, and $Q$ encodes both the impurity measure and the structural constraints as quadratic penalty terms. The objective function for each potential split jointly captures the variance minimization objective and the logical constraints ensuring valid partitioning of categories. Further details are provided in Appendix~\ref{Appendix: QUBO} in the Online Supplement.

\subsubsection{Iterative QUBO Model}

From~\eqref{eq:proposed_fractional}, the node-splitting optimization problem reduces to
\begin{equation}
\label{eq:min_var}
    \argmin_{\mathcal{N}_{t_L},\, \mathcal{N}_{t_R}} 
    \left\{ N_L \operatorname{Var}_L + N_R \operatorname{Var}_R \right\},
\end{equation}
\begin{equation}
\label{eq:var_def}
\text{where}\; \operatorname{Var}_k = \frac{1}{2N_k ^2}\sum\limits_{i,j \in \mathcal{N}_{t_k}} (Y_i - Y_j)^2,\;  \text{for}\; k\in\{L,R\}.  
\end{equation}

Let $\mathcal{S}$ denote the set of admissible splits $\delta$ of node 
$\mathcal{N}_t$ into $\mathcal{N}_{t_L}$ and $\mathcal{N}_{t_R}$, and\\
$\mathcal{R}_\delta = N_L \operatorname{Var}_L + N_R \operatorname{Var}_R$ the splitting cost associated with split $\delta$. Each admissible split $\delta$ corresponds to a binary vector $\mathbf{q}$ of class assignment variables $q_\alpha$, so that $\mathcal{R}_\delta = \mathcal{R}(\mathbf{q})$ for some function \\ $\mathcal{R}: \mathcal{Q}^* \rightarrow \mathbb{R}$. Lemma~\ref{Lemma:binary-sum-var} and Lemma~\ref{lemma:fractional_decomposition} establish the explicit expression of $\mathbf{q} \mapsto \mathcal{R}(\mathbf{q})$, which constitutes the first step toward casting the problem into the QUBO framework. 

\begin{lemma}[Weighted node variances via binary variables]
\label{Lemma:binary-sum-var}
Let $q_{\alpha}$, $q_{\beta}$ be defined by (\ref{eq:q_alphabeta}) and let $ N_L \operatorname{Var}_L$ and $N_R \operatorname{Var}_R$ be the weighted variances of the left and right child nodes, with $\operatorname{Var}_L$, and $\operatorname{Var}_R$ defined in (\ref{eq:var_def}). Then, we have


\begin{align*}
    N_L \operatorname{Var}_L = \frac{1}{N_L} \sum_{\alpha,\beta} q_\alpha q_\beta 
    V_{\alpha\beta}, 
    \qquad
    N_R \operatorname{Var}_R = \frac{1}{N_R} \sum_{\alpha,\beta} 
    (1-q_\alpha)(1-q_\beta) V_{\alpha\beta},
\end{align*}


\begin{equation*}
 \text{where}\quad   V_{\alpha,\beta} = \frac{1}{2} \sum_{i \in C_\alpha} \sum_{j \in C_\beta} 
    (Y_i - Y_j)^2,
    \quad
    N_L = \sum_\alpha q_\alpha N_\alpha,
    \quad 
    N_R = \sum_\alpha (1-q_\alpha) N_\alpha,
\end{equation*}
and $N_\alpha$ denotes the number of observations belonging to category $C_\alpha$.
\end{lemma}

Building on Lemma~\ref{Lemma:binary-sum-var}, Lemma~\ref{lemma:fractional_decomposition} shows that 
$\mathcal{R}_\delta = \mathcal{R}(\mathbf{q})$ can be expressed as a fractional program in the sense of Lemma~\ref{lemma:dinkelbach}, thereby identifying the numerator $\mathbf{n(q)}$ and denominator $\mathbf{d(q)}$ required to apply Dinkelbach's parametric reformulation.

\begin{lemma}[Fractional decomposition of the cost function]
\label{lemma:fractional_decomposition}
The node-splitting cost function $\mathcal{R}_\delta = \mathcal{R}(\mathbf{q)}$ admits the fractional representation $ \mathcal{R}(\mathbf{q}) = \mathbf{n}(\mathbf{q})/\mathbf{d}(\mathbf{q})$ for $\mathbf{q} \in \mathcal{Q}^{*}$, where


\begin{align*}
\mathbf{d}(\mathbf{q}) &= N_L N_R 
    = \sum_\alpha N_S N_\alpha q_\alpha 
    - \sum_{\alpha,\beta} N_\alpha N_\beta q_\alpha q_\beta, \; \text{and}\\
    \mathbf{n}(\mathbf{q}) &= \sum_{\alpha,\beta} q_\alpha q_\beta 
    \left(N_S V_{\alpha,\beta} - N_\beta\sum_\gamma V_{\alpha,\gamma} 
    - N_\alpha\sum_\gamma V_{\gamma,\beta}\right) 
    + N_S ^2\operatorname{Var}_S \sum_\alpha N_\alpha q_\alpha.
\end{align*}
\end{lemma}


By Lemma~\ref{lemma:dinkelbach}, minimizing $\mathcal{R}(\mathbf{q})$ over $\mathcal{Q}^{*}$ is equivalent to solving 
$\min_{\mathbf{q} \in \mathcal{Q}^{*}} F(\lambda, \mathbf{q})$ 
for a given $\lambda \geq 0$. By Lemma~\ref{lemma:optim_condition}, 
the optimal pair $(\lambda^*, \mathbf{q}^*)$ is characterized by 
$F(\lambda^*, \mathbf{q}^*) = 0$. The parameter $\lambda$ is iteratively updated via  $\lambda_{n+1} = \mathbf{n}(\mathbf{q}^*_n)/\mathbf{d}(\mathbf{q}^*_n)$, where \\$\mathbf{q}^*_n \in \arg\min_{\mathbf{q} \in \mathcal{Q}^{*}}\; F(\lambda_n, \mathbf{q})$.

The function $F$ is defined over the full binary space $\{0,1\}^M$, which includes the trivial vectors ($\mathbf{q} \in \{\mathbf{0}, \mathbf{1}\}$) satisfying $\mathbf{d(q)} = 0$. Consequently, the constraint $\mathbf{d(q)} > 0$ required by Lemma~\ref{lemma:dinkelbach} is not automatically satisfied over $\{0,1\}^M$. Enforcing this constraint explicitly would require introducing slack variables, thereby increasing the number of binary variables in the QUBO formulation, as commonly done in the literature. To avoid this, we instead solve the minimization of $F$ over the full space $\{0,1\}^M$ without enforcing $\mathbf{d(q)} \neq 0$. This choice raises two issues: the definition of the update rule for $\lambda$ when $\mathbf{d(q)} = 0$, and the design of an appropriate convergence criterion that excludes trivial solutions. To address the first issue, we adopt the standard statistical convention that both the number of observations and the variance of an empty node are zero. Lemma~\ref{lemma:lambda_update} defines a valid update rule for $\lambda$ at each iteration.

\begin{lemma}[Update rule for the parameter $\lambda$]
\label{lemma:lambda_update}
At iteration $n \geq 1$, let\\ $\mathbf{q}^*_n \in \arg\min_{\mathbf{q} \in 
\mathcal{Q}^*} F(\lambda_n, \mathbf{q})$ be the optimal solution. Then,
\begin{enumerate}[label=\roman*)]
    \item If $\mathbf{d}(\mathbf{q}^*_n) > 0$, Lemma~\ref{lemma:optim_condition} yields
        $\lambda_{n+1}
        = \mathbf{n}(\mathbf{q}^*_n)/\mathbf{d}(\mathbf{q}^*_n)
        = N_L \operatorname{Var}_{L} + N_R \operatorname{Var}_{R}.$
    \item If $\mathbf{d}(\mathbf{q}^*_n) = 0$, we have $\mathbf{q}^*_n \in \{\mathbf{0}, \mathbf{1}\}$ and we set 
    $\lambda_{n+1} = N_S \operatorname{Var}_S$.
\end{enumerate}
\end{lemma}

\begin{theorem}[Iterative QUBO formulation]
\label{theorem:iterative_qubo}
At iteration $n \geq 1$, the node-splitting optimization problem reduces to the QUBO problem


\begin{align*}
    \argmin_{\mathbf{q} \in \{0,1\}^M} F(\lambda_n, \mathbf{q}) &= 
    \argmin_{\mathbf{q} \in \{0,1\}^M}
    \sum_{\alpha, \beta} Q_{\alpha, \beta} \, q_\alpha q_\beta
    + \sum_{\alpha} L_\alpha \, q_\alpha
    = \argmin_{\mathbf{q} \in \{0,1\}^M} \mathbf{q}^\top H \mathbf{q},
\end{align*}


with coefficients


\begin{align*}
\Biggl\{
\begin{array}{ll}
    Q_{\alpha,\beta}
    &= N_S V_{\alpha,\beta}
       - N_\beta \sum_{\gamma} V_{\alpha,\gamma}
       - N_\alpha \sum_{\gamma} V_{\gamma,\beta}
       + \lambda_n N_\alpha N_\beta, \\[0.6em]
    L_\alpha
    &= \left( N_S ^2 \operatorname{Var}_S - N_S \lambda_n \right) N_\alpha,
\end{array}
\Biggr.
\end{align*}


\noindent defining the entries of the symmetric QUBO matrix $H$, and where $\lambda_n$ is updated according to Lemma~\ref{lemma:lambda_update} with initial value 
$\lambda_0 \geq 0$.
\end{theorem}

This formulation allows the node-splitting problem to be solved efficiently using QUBO solvers or quantum-inspired optimization techniques, while fully preserving the statistical structure of the original variance-based objective.

\subsubsection{Properties of $F$ and Convergence Analysis}

We now characterize the structural properties of the objective function $F$ induced by the QUBO formulation.

\begin{lemma}[Properties of $F$]
\label{F_properties}
For any $\lambda \geq 0$ and any binary vector $\mathbf{q}$, the following properties hold:
\begin{enumerate}[label=\roman*)]
    \item \textbf{Binary-flip symmetry:} $F(\lambda, \mathbf{q})  = F(\lambda, \bar{\mathbf{q}})$, 
    where $\bar{\mathbf{q}} = \mathbf{1} - \mathbf{q}$.
    \item \textbf{Trivial solutions:} $F(\lambda, \mathbf{q_{0}}) = F(\lambda, \bar{\mathbf{q}}_{0}) = 0$ 
    for $\mathbf{q_0} = \mathbf{0}$ and $\bar{\mathbf{q}}_0 = \mathbf{1}$.
\end{enumerate}
\end{lemma}

The binary-flip symmetry property shows that any vector $\mathbf{q}$ and its binary complement $\bar{\mathbf{q}}$, obtained by flipping all binary entries, achieve the same objective value. In addition, the trivial solutions $\mathbf{q}_0$ and $\bar{\mathbf{q}}_{0}$ always satisfy $F = 0$, for any $\lambda$. Since the optimal pair $(\lambda^*, \mathbf{q}^*)$ also satisfies $F(\lambda^*, \mathbf{q}^*) = 0$ (Lemma~\ref{lemma:optim_condition}), the condition $F=0$ alone does not guarantee convergence. A reinforced convergence criterion that explicitly excludes trivial solutions is therefore required. Lemma~\ref{lemma:convergence_equivalence} characterizes all zeros of $F$, showing that $F= 0$ corresponds either to a trivial solution or to the optimal solution pair.

\begin{lemma}[Convergence equivalence]
\label{lemma:convergence_equivalence}
Let $(\lambda^*, \mathbf{q}^*)$ be the global optimum, and let \\$\lambda_n = \mathbf{n}(\mathbf{q}^*_{n-1})/
\mathbf{d}(\mathbf{q}^*_{n-1})$ be the parameter updated at iteration $n-1$, 
where $\mathbf{q}^*_{n-1}$ is the optimal solution at the previous iteration.
The condition $F(\lambda_n, \mathbf{q}^*_n) = 0$ holds if and only if 
$\mathbf{q}^*_n \in \{\mathbf{0}, \mathbf{1}\}$ (trivial solution) or 
$(\lambda_n, \mathbf{q}^*_n) = (\lambda^*, \mathbf{q}^*)$.
\end{lemma}

In light of Lemma~\ref{lemma:convergence_equivalence}, Theorem~\ref{theorem:convergence} establishes a reinforced finite-time convergence criterion that excludes trivial solutions.

\begin{theorem}
\label{theorem:convergence}
If there exists $k \geq 1$ such that $F(\lambda_k, \mathbf{q}^*_k) = 0$ with $\mathbf{q}^*_k \in \mathcal{Q}^*$, then the algorithm converges in finite time and $(\lambda_k, \mathbf{q}^*_k) = (\lambda^*, \mathbf{q}^*)$.
\end{theorem}

In practice, given a tolerance $\delta > 0$, the algorithm stops at iteration $k$ if $|F(\lambda_k, \mathbf{q}_k)| \leq \delta$ with $\mathbf{q}_k \in \mathcal{Q}^*$; otherwise, $\lambda$ is updated according to Lemma~\ref{lemma:lambda_update}.

\begin{remark}
When using a QUBO solver (for instance, \texttt{Gurobi}), trivial solutions are discarded at the algorithmic level in accordance with Theorem~\ref{theorem:convergence}. This strategy enforces 
non-triviality without introducing additional slack variables into 
the QUBO formulation.
\end{remark}

We conclude this analysis by establishing fundamental properties of the parameter $\lambda$ and the objective function $F(\lambda, \mathbf{q})$. These properties serve two practical purposes: they provide a theoretically grounded initialization strategy by showing that $\lambda_0 = N_S \operatorname{Var}_S$ is a valid and natural upper bound for $\lambda^*$, and they guarantee that, under this initialization, the sequence $(\lambda_n)$ decreases monotonically and converges to the optimal parameter $\lambda^*$.

\begin{lemma}[Properties of $\lambda$ and $F(\lambda, \mathbf{q})$]
\label{lemma:lambda_properties}
Let $\mathbf{q^*}$ denote the optimal non-trivial solution and $\lambda^*$, the associated optimal parameter. Then,
\begin{enumerate}[label=\roman*)]
    \item $\lambda^* \leq N_S \operatorname{Var}_S$, where $N_S$ and 
    $\operatorname{Var}_S$ denote the number of observations and the variance of the parent node, respectively.
    \item For any $\lambda < \lambda^*$ and any non-trivial $\mathbf{q} \in \mathcal{Q^*}$, $F(\lambda, \mathbf{q})  > 0$ holds, and the minimum is attained by a trivial solution.
    
    \item For any $\lambda \geq N_S \operatorname{Var}_S$, there exists a non-trivial $\mathbf{q}$ such that $F(\lambda, \mathbf{q}) \le 0$, and the minimum is attained by a non-trivial solution.
    \item If $\lambda_0 \geq N_S \operatorname{Var}_S$, then the sequence $(\lambda_n)$ generated by the update rule of Lemma~\ref{lemma:lambda_update} is monotonically decreasing and converges to $\lambda^*$.
\end{enumerate}
\end{lemma}

\subsection{Implementation}
This section describes the practical implementation of the proposed methodology, including the computational environment, optimization solvers, and the main algorithms used to construct and prune the decision trees.

\subsubsection{Optimization Solvers and Computational Settings}
We employed multiple optimization solvers encompassing both exact and heuristic strategies to ensure robust solution procedures. In particular, we utilized the Python interfaces of \texttt{Gurobi}, \texttt{neal}, and \texttt{dimod} for QUBO modeling and solution analysis. Final implementation and numerical experiments were conducted using the academic version of \texttt{Gurobi solver (academic license for non-commercial use, license ID 2642785)}. The computational environment is fully reproducible. All experiments were executed in a \texttt{Python 3.12.10} with standard QUBO-compatible libraries. Solver configurations were fixed across experiments to guarantee numerical consistency and result reproducibility. Source code and configuration files are available upon reasonable request, enabling independent verification of the reported findings. The proposed implementation integrates seamlessly with existing Python frameworks for QUBO modeling and optimization, supporting methodological transparency and computational reproducibility.

\subsubsection{QUBO tree construction process}
We propose a new version of the CART algorithm that integrates QUBO into the tree construction process. In this approach, the splitting strategy depends on the nature of the predictors: for categorical variables, the search for the optimal split is formulated and solved as a QUBO optimization problem, whereas for continuous and ordinal variables, the algorithm retains the standard splitting rule implemented in \texttt{rpart}, based on an exhaustive search over candidate thresholds. Except for this specific treatment of categorical variables, the overall tree structure and stopping criteria largely follow those of the conventional \texttt{rpart} package. The formal representation of the algorithm is presented in Algorithm~\ref{alg:cart_qubo} (see Appendix~\ref{Appendix:algo build tree} in the Online Supplement).

\subsubsection{Pruning and Final Tree Selection}

Following the construction of the maximal tree $\mathcal{A}_{\max}$, CART applies the standard cost--complexity pruning procedure to control model complexity \citep{breiman1984cart}. This approach selects a sequence of nested subtrees by minimizing the cost--complexity criterion
\begin{equation}
C_{\alpha}(\mathcal{A}) = R_{\mathcal{A}} + \alpha\,|\mathcal{A}|, \qquad \alpha \ge 0,
\end{equation}
which balances empirical risk and tree size. As shown by \citet{breiman1984cart}, this yields a finite set of candidate subtrees corresponding to increasing values of the regularization parameter $\alpha$. The final tree is selected using a validation-based model selection strategy, where the subtree minimizing the prediction error on a validation set is retained, following standard practice in statistical learning \citep{hastie2009elements, james2013introduction}. The practical implementation of the cost--complexity pruning procedure used in this work is detailed in Algorithm~\ref{alg:cost_complexity} in Appendix~\ref{Appendix:algo build tree} in the Online Supplement. The test set is used exclusively for final performance evaluation.


\section{Numerical Results}

We present an empirical evaluation of our QUBO formulation for the CART splitting optimization problem. Our approach is assessed across multiple datasets in a comprehensive comparative study, including both simulated data and real-world insurance data. We focus on insurance datasets, as regression tree models have been successfully applied to individual loss reserving \citep{lopez2016tree, janouvsek2025bagging} and credibility modeling \citep{diao2019regression}. The main objectives of this empirical evaluation are as follows:
\begin{enumerate}[label=\Roman*.]
\setlength{\itemsep}{-0.25em}
    \item Compare the optimal solutions obtained from the QUBO solvers with those from exact optimization for categorical variables.
    \item Analyze the convergence behavior of the iterative QUBO procedure for the split optimization problem.

    \item Evaluate the predictive performance of the QUBO-based regression tree in comparison with classical CART algorithms under both pre-pruning and post-pruning strategies.
\end{enumerate}

\subsection{Dataset Overview}

Each dataset used in this work is described by its source, simulation procedure, or generation process. This information enables the reproduction of the results.

\subsubsection{Simulated Data}

We generated one simple synthetic dataset and three synthetic insurance claim datasets of varying sizes. 
The dataset \textbf{df\_tricky\_means} contains $5{,}000$ observations and consists of a single categorical explanatory variable and one response variable. 
The response variable represents the automobile price (\textbf{price}), while the explanatory variable corresponds to the vehicle color (\textbf{color}). We also generated the dataset \textbf{datagen}, with sample sizes of $10{,}000$ and $50{,}000$ observations, and the dataset \textbf{df}, which contains $20{,}000$ observations. Both  \textbf{df} and \textbf{datagen} include multiple explanatory variables of mixed types (continuous, categorical, and binary), namely \textbf{Brand}, \textbf{Color}, \textbf{Mileage\_km}, and \textbf{HasClaim}. The response variable is the claim amount, denoted by \textbf{ClaimAmount}. The data generation procedures are described in Algorithms~\ref{alg:datagen_generation} and \ref{alg:df_generation} (see Appendix~\ref{Appendix: datasets} in the Online Supplement).

\subsubsection{Real-World Datasets}

To assess the practical performance and robustness of the proposed methodology, we conduct experiments on three real-world datasets drawn from the real estate and insurance domains. The \textbf{Ames Housing} dataset, available in the \texttt{sklearn.datasets} (\url{https://github.com/INRIA/scikit-learn-mooc/blob/main/datasets/}) Python library, is a well-known benchmark dataset for regression problems in real estate valuation. In addition, we consider two insurance datasets from the \texttt{CASdatasets}
(\url{https://github.com/dutangc/CASdatasets/blob/master/data/})
 R package: \textbf{freMPL}, a French Motor Private Line insurance dataset, and \textbf{ausprivauto}, an Australian Private Auto Line insurance dataset. These datasets provide heterogeneous feature structures and realistic claim distributions, allowing for a comprehensive empirical evaluation of the proposed approach. Table~\ref{tab:datasets_variables} shows a subset of variables selected for the analysis in each dataset.

\subsection{Results and Analysis}
We present the empirical results together with an analysis of the main findings. We first report the optimization results for both simulated and real-world datasets. Subsequently, we provide a comparative analysis of the predictive performance of the decision tree models across datasets.

\subsubsection{Optimization Performance}

We evaluate the performance of the proposed QUBO-based approach on both simulated datasets and real-world insurance datasets. The objective is to identify the optimal binary partition that minimizes the node-splitting criterion and to analyze the convergence behavior of the algorithm (see objectives I. and II.). Across all datasets considered, the QUBO formulation solved with the \texttt{Gurobi solver} yields optimal categorical splits that are identical to those obtained by exhaustive search, validating the exactness of the proposed approach.

We further analyze the convergence behavior of the iterative QUBO procedure. In all tables, the optimal binary solution $\mathbf{q^{\star}}$ (in bold) corresponds to the optimal split. Convergence results are presented for the simulated dataset \textbf{datagen} (Table~\ref{tab:datagen_convergence}) and 
the real-world dataset \textbf{freMPL} (Table~\ref{tab:frempl_convergence}); remaining results are reported in Appendix~\ref{Appendix: Optim results} (Tables~\ref{tab:tricky_color_iterations}, \ref{tab:ames_convergence}, and~\ref{tab:ausprivauto_conv}) in the Online Supplement.

\begin{table}[ht]
\centering
\scriptsize
\caption{Convergence of the QUBO algorithm on the \texttt{datagen} dataset.}
\label{tab:datagen_convergence}
\begin{tabular}{l l l l l l}
\toprule
Variable & Iteration & $\lambda_{\text{initial}}$ & Binary vector & Score & $\lambda_{\text{final}}$ \\ 
\midrule
\multirow{2}{*}{Color}
& 1 & $3.6628\times10^{11}$ & $(0,0,0,0,1,1)$ & $3.6612 \times10^{7}$ & $3.6612\times10^{11}$ \\
& 2 & $3.6612\times10^{11}$ & $\mathbf{(0,0,0,0,1,1)}$ & $3.6612 \times10^{7}$ & $3.6612\times10^{11}$ \\ 
\midrule

\multirow{2}{*}{Brand}
& 1 & $3.6628\times10^{11}$ & $(0,0,1,1,1,1,0,1,1,1)$ & $3.6495 \times10^{7}$ & $3.6495\times10^{11}$ \\
& 2 & $3.6495\times10^{11}$ & $\mathbf{(0,0,1,1,1,1,0,1,1,1)}$ & $3.6495 \times10^{7}$ & $3.6495\times10^{11}$ \\ 
\bottomrule
\end{tabular}

\footnotesize
\scriptsize
\textbf{Binary vector encoding:}

\textit{Color:}
$(q_1,\dots,q_6)=(\text{Black},\text{Blue},\text{Gray},\text{Green},\text{Red},\text{White})$. 

\textit{Brand:} $(q_1,\dots,q_{10}) = $ (\text{Audi}, \text{BMW}, \text{Ford}, \text{Honda}, \text{Hyundai}, \text{Kia}, \text{Mercedes}, \text{Nissan},\\ \text{Toyota}, \text{Volkswagen}).
\end{table}


\begin{table}[ht]
\centering
\scriptsize
\caption{Convergence of the QUBO algorithm on the \texttt{freMPL} dataset.}
\label{tab:frempl_convergence}
\begin{tabular}{c c c l c c}
\toprule
Variable & Iteration & $\lambda_{\text{initial}}$ & Binary vector & Score & $\lambda_{\text{final}}$ \\ 
\midrule

\multirow{5}{*}{VehBody}
& 1 & 0 & (0,0,0,0,0,0,0,0,0) & $5.3836 \times 10^{8}$ & $1.7577\times 10^{12}$ \\
& 2 & $1.7577\times 10^{12}$ & (0,1,1,1,1,0,1,0,0) & $5.3481 \times 10^{8}$ & $1.7462\times 10^{12}$ \\
& 3 & $1.7462\times 10^{12}$ & (1,1,1,0,0,1,0,1,1) & $5.3376 \times 10^{8}$ & $1.7427\times 10^{12}$ \\
& 4 &$ 1.7427\times 10^{12}$ & (0,0,0,1,0,0,0,0,0) & $5.3280 \times 10^{8}$ & $1.7396\times 10^{12}$ \\
& 5 & $1.7396\times 10^{12}$ & $\mathbf{(1,1,1,0,1,1,1,1,1)}$ & $5.3280 \times 10^{8}$ & $1.7396\times 10^{12}$ \\ 
\midrule

\multirow{4}{*}{VehClass}
& 1 & 0 & (0,0,0,0,0,0) & $5.3836 \times 10^{8}$ & $1.7577\times 10^{12}$ \\
& 2 & $1.7577\times 10^{12}$ & (0,1,1,0,1,1) & $5.3574 \times 10^{8}$ & $1.7492\times 10^{12}$ \\
& 3 & $1.7492\times 10^{12}$ & (1,1,1,0,1,1) & $5.3549 \times 10^{8}$ & $1.7484\times 10^{12}$ \\
& 4 & $1.7484\times 10^{12}$ & $\mathbf{(1,1,1,0,1,1)}$ & $5.3549 \times 10^{8}$ & $1.7484\times 10^{12}$ \\ 
\bottomrule
\end{tabular}

\footnotesize{
\scriptsize
\textbf{Binary vector encoding:}  

\textit{VehBody :} $(q_1, q_2, ..., q_{10})$ = $(\text{bus, cabriolet, coupe, microvan, other microvan, }$ $\text{sedan, sport utility vehicle,  station wagon, van})$. 

\textit{VehClass :} $(q_1, q_2, ..., q_{6})$ =  $(\text{0, A, B, H, M1, M2})$.
}
\end{table}

These results show that the QUBO formulation exactly reproduces the optimal split obtained via exhaustive search and converges to the non-trivial optimal solution. The reinforced convergence criterion effectively prevents convergence to trivial partitions. Overall, this experiment confirms the robustness and reliability of the proposed iterative QUBO approach for categorical splitting. The QUBO method converges within only a few iterations. At the final iteration, multiple candidate solutions satisfy the tolerance criterion with nearly zero objective value. The reinforced convergence criterion successfully discards trivial solutions, thereby ensuring the identification of the correct non-trivial optimum. Moreover, the evolution of the penalty parameter $\lambda$ decreases monotonically until convergence.

\subsubsection{Predictive Performance Evaluation}

We compare the predictive performance of QUBO-based decision tree with classical CART implementations, namely the \textit{scikit-learn} (SK) library ({\url{https://scikit-learn.org/stable/}) and the RPART algorithm, on simulated and real-world insurance datasets (see Objective III). Predictive accuracy is evaluated using the \textbf{mean squared error (MSE)}, reported relative to the RPART baseline (\textbf{Rel.\ MSE}). Model performance is further assessed through a comparative analysis of pre-pruning and post-pruning strategies based on a test-sample evaluation.

\subsubsection*{Pre-pruning}
For this analysis, we control the tree complexity by fixing the tree depth or by imposing a minimum number of observations per leaf, thereby forcing the tree to stop growing at an earlier stage than the fully grown tree. In our experiments, the tree depth is set to $5$, although other depth values may also be considered. We then compare the prediction errors of the trees obtained using each method. The corresponding results are reported in Table~\ref{tab:depth5_mse_comparison}.

\begin{table}[ht!]
\centering
\scriptsize
\caption{Performance comparison for depth--5 decision trees}
\label{tab:depth5_mse_comparison}
\begin{tabular}{llccccc}
\toprule
\textbf{Dataset} & \textbf{Method} & \textbf{Leaves} &\textbf{ MSE} & \textbf{Rel. MSE (\%)} \\
\midrule

\multirow{3}{*}{df\_tricky\_means}
 & QUBO & 6 & $1.1492\times10^{7}$ & 0.000 \\
 & SK   & 6 & $1.1492\times10^{7}$ &  0.000 \\
 & RPART     & 6 & $1.1492\times10^{7}$  & 0.000 \\
\midrule

\multirow{3}{*}{datagen}
 & QUBO & 32 & $5.1245\times10^{6}$ &  0.000 \\
 & SK      & 32 & $5.1268\times10^{6}$ &  $+0.045$ \\
 & RPART     & 32 & $5.1245\times10^{6}$ &  0.000 \\
\midrule

\multirow{3}{*}{ames}
 & QUBO & 22 & $4.7992\times10^{-1}$ &  $-0.000$ \\
 & SK      & 23 & $4.8327\times10^{-1}$ &  $+0.698$ \\
 & RPART     & 22 & $4.7992\times10^{-1}$  & 0.000 \\
\midrule

\multirow{3}{*}{freMPL}
 & QUBO & 28 & $2.9797\times10^{8}$ & $+1.759$ \\
 & SK      & 26 & $2.9804\times10^{8}$ & $+1.783$ \\
 & RPART     & 29 & $2.9282\times10^{8}$ & 0.000 \\
\midrule

\multirow{3}{*}{ausprivauto}
 & QUBO & 31 & $1.4410\times10^{10}$  & 0.000 \\
 & SK      & 28 & $1.5491\times10^{10}$  & $+7.495$ \\
 & RPART     & 31 & $1.4410\times10^{10}$ & 0.000 \\
\bottomrule
\end{tabular}
\end{table}

We observe that the models have nearly the same number of leaves and 
produce identical predictions at this stage. This provides empirical 
evidence that the QUBO-based decision tree model is consistent with 
the classical CART formulation. Furthermore, the characteristics of the QUBO and RPART trees are highly similar, almost identical, which is expected since both models are constructed using the same approach.

\subsubsection*{Post-pruning: Test sample selection}
To refine our models, we employ the classical cost--complexity pruning method \citep{breiman1984cart, bradford1998pruning} to identify the optimal subtree that minimizes the prediction error. This process effectively manages the trade-off between model complexity and generalization. Specifically, we use a test sample selection approach, which involves partitioning the dataset into three distinct subsets: 50\% for training, 25\% for validation, and 25\% for testing. The procedure is structured as follows:
\begin{enumerate}
\setlength{\itemsep}{-0.25em}
    \item Construction of the maximal tree using the training set.
    \item Generation of a nested sequence of pruned subtrees along with their corresponding complexity parameters ($\alpha_m$).
    \item Selection of the optimal tree based on the validation set; we identify the subtree that minimizes the validation error within the generated sequence.
    \item Final evaluation, where the selected optimal tree is assessed using the independent test sample.
\end{enumerate}
This methodology requires a sufficiently large dataset to ensure statistical robustness. Given the substantial size of our data, we apply this approach to the \texttt{datagen} dataset ($50{,}000$ observations), as well as to the \texttt{df} and \texttt{freMPL} datasets. For the \texttt{freMPL}, the variable \textbf{VehBody} contains a limited number of observations in certain categories. Moreover, the variable \textbf{VehPrice} consists of 26 classes, some of which are sparsely represented. To mitigate issues related to low representativeness, these variables were initially excluded from the model and subsequently recoded in order to reintegrate them into the modeling framework. For each tree configuration (root model, optimal post-pruned model, and maximal tree), we report the errors observed across the three samples (training, validation, and test). In addition, we present the relative difference (\%) between the test error of each method and the baseline error obtained using the RPART method. The results are detailed in Tables~\ref{tab:df_test_error_comparison}, \ref{tab:datagen_test_error_comparison},
\ref{tab:freMPL_test_error_comparison_reduce_var} and
\ref{tab:freMPL_test_error_comparison_recod_var}.

\begin{table}[!ht]
\centering
\scriptsize
\caption{Comparison of prediction errors for dataset \texttt{df (20,000 observations)}.}
\label{tab:df_test_error_comparison}
\begin{tabular}{llcccc}
\toprule
\textbf{Tree type} & \textbf{Method} &
\makecell{\textbf{Leaves} \\ \textbf{(Depth)}} &
\makecell{\textbf{Train MSE} \\ \textbf{(Rel. MSE \%)}} &
\makecell{\textbf{Validation MSE} \\ \textbf{(Rel. MSE \%)}} &
\makecell{\textbf{Test MSE} \\ \textbf{(Rel. MSE \%)}} \\
\midrule
\multirow{3}{*}{Root Tree}
& RPART & \makecell{1 \\ (0)}
   & $1.8065\times10^{7}$
   & $1.8345\times10^{7}$
   & $1.7861\times10^{7}$ \\
 & QUBO  & \makecell{1 \\ (0)}
   & \makecell{$1.8065\times10^{7}$ \\ (0.00)}
   & \makecell{$1.8345\times10^{7}$ \\ (0.00)}
   & \makecell{$1.7861\times10^{7}$ \\ (0.00)} \\[6pt]
 & SK    & \makecell{1 \\ (0)}
   & \makecell{$1.8065\times10^{7}$ \\ (0.00)}
   & \makecell{$1.8345\times10^{7}$ \\ (0.00)}
   & \makecell{$1.7861\times10^{7}$ \\ (0.00)} \\[6pt]
 \midrule
\multirow{3}{*}{\makecell[l]{Validation \\ Best Tree}}
& RPART & \makecell{10 \\ (5)}
   & $1.2515\times10^{6}$
   & $1.1618\times10^{6}$
   & $1.2061\times10^{6}$ \\
 & QUBO  & \makecell{10 \\ (5)}
   & \makecell{$1.2515\times10^{6}$ \\ (0.00)}
   & \makecell{$1.1618\times10^{6}$ \\ (0.00)}
   & \makecell{$1.2061\times10^{6}$ \\ (0.00)} \\[6pt]
 & SK    & \makecell{17 \\ (7)}
   & \makecell{$1.2613\times10^{6}$ \\ (+0.78)}
   & \makecell{$1.2006\times10^{6}$ \\ (+3.34)}
   & \makecell{$1.2215\times10^{6}$ \\ (+1.28)} \\[6pt]
\midrule
\multirow{3}{*}{\makecell[l]{Test \\ Best Tree}}
& RPART & \makecell{17 \\ (7)}
   & $1.2323\times10^{6}$
   & $1.1790\times10^{6}$
   & $1.2021\times10^{6}$ \\
 & QUBO  & \makecell{17 \\ (7)}
   & \makecell{$1.2323\times10^{6}$ \\ (0.00)}
   & \makecell{$1.1790\times10^{6}$ \\ (0.00)}
   & \makecell{$1.2021\times10^{6}$ \\ (0.00)} \\[6pt]
 & SK    & \makecell{21 \\ (7)}
   & \makecell{$1.2423\times10^{6}$ \\ (+0.81)}
   & \makecell{$1.2044\times10^{6}$ \\ (+2.15)}
   & \makecell{$1.2100\times10^{6}$ \\ (+0.66)} \\[6pt]
\midrule
\multirow{3}{*}{Max Tree}
 & RPART & \makecell{3706 \\ (30)}
   & $2.1598\times10^{4}$
   & $2.4576\times10^{6}$
   & $2.3814\times10^{6}$ \\
 & QUBO  & \makecell{3833 \\ (39)}
   & \makecell{$1.8579\times10^{3}$ \\ ($-91.40$)}
   & \makecell{$2.4708\times10^{6}$ \\ (+0.54)}
   & \makecell{$2.4345\times10^{6}$ \\ (+2.23)} \\[6pt]
 & SK    & \makecell{3831 \\ (42)}
   & \makecell{$1.8579\times10^{3}$ \\ ($-91.40$)}
   & \makecell{$2.5434\times10^{6}$ \\ (+3.49)}
   & \makecell{$2.4626\times10^{6}$ \\ (+3.37)} \\[6pt]
\bottomrule
\end{tabular}
\end{table}


\begin{table}[!ht]
\centering
\scriptsize
\caption{Comparison of prediction errors for dataset \texttt{datagen (50,000 observations)}.}
\label{tab:datagen_test_error_comparison}
\begin{tabular}{llcccc}
\toprule
\textbf{Tree type} & \textbf{Method} &
\makecell{\textbf{Leaves} \\ \textbf{(Depth)}} &
\makecell{\textbf{Train MSE} \\ \textbf{(Rel. MSE \%)}} &
\makecell{\textbf{Validation MSE} \\ \textbf{(Rel. MSE \%)}} &
\makecell{\textbf{Test MSE} \\ \textbf{(Rel. MSE \%)}} \\
\midrule
\multirow{3}{*}{Root Tree}
& RPART & \makecell{1 \\ (0)}
   & $4.4229\times10^{7}$
   & $3.8042\times10^{7}$
   & $6.4666\times10^{7}$ \\
 & QUBO  & \makecell{1 \\ (0)}
   & \makecell{$4.4229\times10^{7}$ \\ (0.00)}
   & \makecell{$3.8042\times10^{7}$ \\ (0.00)}
   & \makecell{$6.4666\times10^{7}$ \\ (0.00)} \\[6pt]
 & SK    & \makecell{1 \\ (0)}
   & \makecell{$4.4229\times10^{7}$ \\ (0.00)}
   & \makecell{$3.8042\times10^{7}$ \\ (0.00)}
   & \makecell{$6.4666\times10^{7}$ \\ (0.00)} \\[6pt]
\midrule
\multirow{3}{*}{\makecell[l]{Validation \\ Best Tree}}
& RPART & \makecell{2 \\ (1)}
   & $4.0574\times10^{7}$
   & $3.4853\times10^{7}$
   & $6.1032\times10^{7}$ \\
 & QUBO  & \makecell{2 \\ (1)}
   & \makecell{$4.0574\times10^{7}$ \\ (0.00)}
   & \makecell{$3.4853\times10^{7}$ \\ (0.00)}
   & \makecell{$6.1032\times10^{7}$ \\ (0.00)} \\[6pt]
 & SK    & \makecell{2 \\ (1)}
   & \makecell{$4.0574\times10^{7}$ \\ (0.00)}
   & \makecell{$3.4853\times10^{7}$ \\ (0.00)}
   & \makecell{$6.1032\times10^{7}$ \\ (0.00)} \\[6pt]
\midrule
\multirow{3}{*}{\makecell[l]{Test \\ Best Tree}}
& RPART & \makecell{7 \\ (4)}
   & $3.4872\times10^{7}$
   & $4.3175\times10^{7}$
   & $6.0600\times10^{7}$ \\
 & QUBO  & \makecell{7 \\ (4)}
   & \makecell{$3.4872\times10^{7}$ \\ (0.00)}
   & \makecell{$4.3175\times10^{7}$ \\ (0.00)}
   & \makecell{$6.0600\times10^{7}$ \\ (0.00)} \\[6pt]
 & SK    & \makecell{47 \\ (14)}
   & \makecell{$9.0590\times10^{6}$ \\ ($-74.02$)}
   & \makecell{$8.1514\times10^{7}$ \\ (+88.80)}
   & \makecell{$6.0603\times10^{7}$ \\ (+0.005)} \\[6pt]
\midrule
\multirow{3}{*}{Max Tree}
& RPART & \makecell{5882 \\ (30)}
   & $1.9677\times10^{5}$
   & $7.4373\times10^{7}$
   & $7.5802\times10^{7}$ \\
 & QUBO  & \makecell{5887 \\ (31)}
   & \makecell{$1.9677\times10^{5}$ \\ (0.00)}
   & \makecell{$7.4393\times10^{7}$ \\ (+0.03)}
   & \makecell{$7.6084\times10^{7}$ \\ (+0.37)} \\[6pt]
 & SK    & \makecell{5885 \\ (40)}
   & \makecell{$1.9677\times10^{5}$ \\ (0.00)}
   & \makecell{$8.6185\times10^{7}$ \\ (+15.88)}
   & \makecell{$7.1146\times10^{7}$ \\ ($-6.14$)} \\[6pt]
\bottomrule
\end{tabular}
\end{table}

\begin{table}[!ht]
\centering
\scriptsize
\caption{Comparison of prediction errors for dataset \texttt{freMPL (reduced variables)}.}
\label{tab:freMPL_test_error_comparison_reduce_var}
\begin{tabular}{llcccc}
\toprule
\textbf{Tree type} & \textbf{Method} &
\makecell{\textbf{Leaves} \\ \textbf{(Depth)}} &
\makecell{\textbf{Train MSE} \\ \textbf{(Rel. MSE \%)}} &
\makecell{\textbf{Validation MSE} \\ \textbf{(Rel. MSE \%)}} &
\makecell{\textbf{Test MSE} \\ \textbf{(Rel. MSE \%)}} \\
\midrule
\multirow{3}{*}{Root Tree}
& RPART & \makecell{1 \\ (0)}
   & $5.8849\times10^{8}$
   & $3.1646\times10^{8}$
   & $3.0866\times10^{8}$ \\
 & QUBO  & \makecell{1 \\ (0)}
   & \makecell{$5.8849\times10^{8}$ \\ (0.00)}
   & \makecell{$3.1646\times10^{8}$ \\ (0.00)}
   & \makecell{$3.0866\times10^{8}$ \\ (0.00)} \\[6pt]
 & SK    & \makecell{1 \\ (0)}
   & \makecell{$5.8849\times10^{8}$ \\ (0.00)}
   & \makecell{$3.1646\times10^{8}$ \\ (0.00)}
   & \makecell{$3.0866\times10^{8}$ \\ (0.00)} \\[6pt]
\midrule
\multirow{3}{*}{Best Tree}
& RPART & \makecell{16 \\ (6)}
   & $3.1813\times10^{8}$
   & $2.2807\times10^{8}$
   & $2.9788\times10^{8}$ \\
 & QUBO  & \makecell{16 \\ (6)}
   & \makecell{$3.1813\times10^{8}$ \\ (0.00)}
   & \makecell{$2.2807\times10^{8}$ \\ (0.00)}
   & \makecell{$2.9788\times10^{8}$ \\ (0.00)} \\[6pt]
 & SK    & \makecell{16 \\ (6)}
   & \makecell{$3.1816\times10^{8}$ \\ (+0.01)}
   & \makecell{$2.2803\times10^{8}$ \\ ($-0.02$)}
   & \makecell{$2.7434\times10^{8}$ \\ ($-7.90$)} \\[6pt]
\midrule
\multirow{3}{*}{Max Tree}
& RPART & \makecell{1792 \\ (23)}
   & $1.1074\times10^{8}$
   & $4.4637\times10^{8}$
   & $3.3955\times10^{8}$ \\
 & QUBO  & \makecell{1792 \\ (23)}
   & \makecell{$1.1074\times10^{8}$ \\ (0.00)}
   & \makecell{$4.4591\times10^{8}$ \\ ($-0.10$)}
   & \makecell{$3.3951\times10^{8}$ \\ ($-0.01$)} \\[6pt]
 & SK    & \makecell{1792 \\ (25)}
   & \makecell{$1.1074\times10^{8}$ \\ (0.00)}
   & \makecell{$4.9962\times10^{8}$ \\ (+11.93)}
   & \makecell{$3.4507\times10^{8}$ \\ (+1.63)} \\[6pt]
\bottomrule
\end{tabular}
\end{table}

\begin{table}[!ht]
\centering
\scriptsize
\caption{Comparison of prediction errors for dataset \texttt{freMPL (recoded variables)}.}
\label{tab:freMPL_test_error_comparison_recod_var}
\begin{tabular}{llcccc}
\toprule
\textbf{Tree type} & \textbf{Method} &
\makecell{\textbf{Leaves} \\ \textbf{(Depth)}} &
\makecell{\textbf{Train MSE} \\ \textbf{(Rel. MSE \%)}} &
\makecell{\textbf{Validation MSE} \\ \textbf{(Rel. MSE \%)}} &
\makecell{\textbf{Test MSE} \\ \textbf{(Rel. MSE \%)}} \\
\midrule
\multirow{2}{*}{Root Tree}
& RPART & \makecell{1 \\ (0)}
   & $5.8849\times10^{8}$
   & $3.1646\times10^{8}$
   & $3.0866\times10^{8}$ \\
 & QUBO  & \makecell{1 \\ (0)}
   & \makecell{$5.8849\times10^{8}$ \\ (0.00)}
   & \makecell{$3.1646\times10^{8}$ \\ (0.00)}
   & \makecell{$3.0866\times10^{8}$ \\ (0.00)} \\[6pt]
\midrule
\multirow{2}{*}{\makecell[l]{Validation \\ Best Tree}}
& RPART & \makecell{56 \\ (14)}
   & $1.0241\times10^{8}$
   & $1.8142\times10^{8}$
   & $2.7794\times10^{8}$ \\
 & QUBO  & \makecell{58 \\ (14)}
   & \makecell{$1.0007\times10^{8}$ \\ ($-2.28$)}
   & \makecell{$1.7860\times10^{8}$ \\ ($-1.55$)}
   & \makecell{$2.7266\times10^{8}$ \\ ($-1.90$)} \\[6pt]
\midrule
\multirow{2}{*}{\makecell[l]{Test \\ Best Tree}}
& RPART & \makecell{8 \\ (6)}
   & $3.2352\times10^{8}$
   & $3.1549\times10^{8}$
   & $1.8356\times10^{8}$ \\
 & QUBO  & \makecell{8 \\ (6)}
   & \makecell{$3.2352\times10^{8}$ \\ (0.00)}
   & \makecell{$3.1549\times10^{8}$ \\ (0.00)}
   & \makecell{$1.8356\times10^{8}$ \\ (0.00)} \\[6pt]
\midrule
\multirow{2}{*}{Max Tree}
& RPART & \makecell{2088 \\ (30)}
   & $3.0626\times10^{7}$
   & $2.1638\times10^{8}$
   & $2.9670\times10^{8}$ \\
 & QUBO  & \makecell{2102 \\ (33)}
   & \makecell{$3.0616\times10^{7}$ \\ ($-0.03$)}
   & \makecell{$2.0934\times10^{8}$ \\ ($-3.25$)}
   & \makecell{$2.8753\times10^{8}$ \\ ($-3.09$)} \\[6pt]
\bottomrule
\end{tabular}
\end{table}

From the analysis of the different tables, it appears that the three models generally produce similar results. More specifically, the optimal trees yield nearly identical predictions across all models and datasets (training, validation, and test), with only minor differences occasionally observed for the \texttt{scikit-learn} (SK) model. Regarding the QUBO and RPART models, the results of the optimal trees are identical across all datasets, and in some cases, the QUBO optimal trees exhibit a slight performance improvement compared to those of RPART.

For the maximal trees, the models produce identical results on the training data but exhibit slight divergences on the validation and test sets, particularly between QUBO and RPART, despite the expectation of similar characteristics. This difference is primarily explained by the fact that the maximal trees do not have the same number of leaves and therefore differ in depth. Indeed, the RPART algorithm is constrained to a maximum depth of 30, whereas the tree based on the QUBO formulation can achieve greater depth and leverage additional information from the data. This explains the occasionally superior performance of the QUBO-based model on the training data. However, by learning the specific characteristics of the training dataset more extensively, the model may exhibit reduced performance on the validation and test sets, a classical manifestation of overfitting. It should also be noted that even when the maximal QUBO and RPART trees are structurally identical, their predictions may differ on the validation and test sets depending on dataset characteristics. This second explanation is further discussed in the following subsection.

\subsubsection{Analysis of QUBO and RPART Maximal Trees Differences\\}

Consider the \texttt{datagen} dataset composed of $10{,}000$ observations. The comparative results of the different models on the training, validation, and test sets are reported in Table~\ref{tab_datagen_10000}.

\begin{table}[ht!]
\centering
\scriptsize
\caption{Comparison of prediction errors for \texttt{datagen ($10{,}000$ observations)} dataset.}
\label{tab_datagen_10000}
\begin{tabular}{llccccc}
\toprule
\textbf{Type} & \textbf{Method} & \textbf{Leaves} &
\makecell{\textbf{Train MSE} \\ \textbf{(Rel. MSE \%)}} &
\makecell{\textbf{Validation MSE} \\ \textbf{(Rel. MSE \%)}} &
\makecell{\textbf{Test MSE} \\ \textbf{(Rel. MSE \%)}} \\
\midrule
Root Tree
 & -- & 1
   & $4.6149\times10^{7}$
   & $3.3576\times10^{7}$
   & $1.1102\times10^{7}$ \\
\midrule
\multirow{2}{*}{Valid Best Tree}
 & RPART & 11
   & $5.0727\times10^{6}$
   & $2.9635\times10^{7}$
   & $8.6766\times10^{6}$ \\[6pt]
 & QUBO  & 11
   & \makecell{$5.0727\times10^{6}$ \\ (0.00)}
   & \makecell{$2.9635\times10^{7}$ \\ (0.00)}
   & \makecell{$8.6766\times10^{6}$ \\ (0.00)} \\
\midrule
\multirow{2}{*}{Test Best Tree}
 & RPART & 16
   & $3.3531\times10^{6}$
   & $3.4122\times10^{7}$
   & $8.5298\times10^{6}$ \\[6pt]
 & QUBO  & 16
   & \makecell{$3.3531\times10^{6}$ \\ (0.00)}
   & \makecell{$3.4122\times10^{7}$ \\ (0.00)}
   & \makecell{$8.5298\times10^{6}$ \\ (0.00)} \\
\midrule
\multirow{2}{*}{Max Tree}
 & RPART & 1168
   & $4.3815\times10^{-2}$
   & $3.7208\times10^{7}$
   & $1.1309\times10^{7}$ \\[6pt]
 & QUBO  & 1168
   & \makecell{$4.3815\times10^{-2}$ \\ (0.00)}
   & \makecell{$3.9055\times10^{7}$ \\ (+4.96)}
   & \makecell{$1.2193\times10^{7}$ \\ (+7.82)} \\
\bottomrule
\end{tabular}
\end{table}

We observe that the two models have the same number of leaves and achieve identical performance on the training data, while differences appear on the validation and test sets. To better understand the origin of these discrepancies, we restricted the maximum tree depth to 5 and 6. The corresponding results are presented in Table~\ref{tab:mse_depth5-6}.

\begin{table}[ht!]
\centering
\scriptsize
\caption{Comparison of regression trees (depth 5-6) for \texttt{datagen ($10{,}000$ observations)} dataset.}
\label{tab:mse_depth5-6}
\begin{tabular}{lcccccc}
\toprule
\textbf{Method} & \textbf{Leaves} & \textbf{Depth} &
\makecell{\textbf{Train MSE} \\ \textbf{(Rel. MSE \%)}} &
\makecell{\textbf{Validation MSE} \\ \textbf{(Rel. MSE \%)}} &
\makecell{\textbf{Test MSE} \\ \textbf{(Rel. MSE \%)}} \\
\midrule

RPART & 17 & 5
  & $1.0435927\times10^{7}$
  & $3.1087463\times10^{7}$
  & $9.801719\times10^{6}$ \\[6pt]
QUBO  & 17 & 5
  & \makecell{$1.0435927\times10^{7}$ \\ (0.00)}
  & \makecell{$3.1087463\times10^{7}$ \\ (0.00)}
  & \makecell{$9.801719\times10^{6}$ \\ (0.00)} \\
\midrule

RPART & 32 & 6
  & $6.853276\times10^{6}$
  & $3.2941964\times10^{7}$
  & $1.1184164\times10^{7}$ \\[6pt]
QUBO  & 32 & 6
  & \makecell{$6.853276\times10^{6}$ \\ (0.00)}
  & \makecell{$3.6037032\times10^{7}$ \\ (+9.40)}
  & \makecell{$1.2624976\times10^{7}$ \\ (+12.88)} \\
\bottomrule
\end{tabular}
\end{table}

Table~\ref{tab:mse_depth5-6} shows that the trees with depth $5$ produced by both models are structurally identical and yield strictly equivalent performance. We then considered trees with depth of $6$, whose results are also reported in Table~\ref{tab:mse_depth5-6}. A difference in predictions is observed on the validation and test datasets. However, inspection of the trees reveals that they share the same leaves as well as identical partitioning rules at the internal nodes. We therefore identified the observations for which the two models produce different predictions; these are reported in Table~\ref{tab:prediction_diff}.

\begin{table}[ht!]
\centering
\scriptsize
\caption{Prediction differences of QUBO and RPART in validation and test datasets.}
\label{tab:prediction_diff}
\begin{tabular}{ccccccccccc}
\toprule
Dataset & ID & Brand & Color & Mileage (km) & HasClaim & ClaimAmount &
\multicolumn{2}{c}{Max Tree} \\
\cline{8-9}
 &  &  &  &  &  &  &
CART & QUBO \\
\midrule
\multirow{2}{*}{Validation}
 & 7106 & Ford & Black & 82228 & 1 & 4939 & 3584.67 & 60488 \\
 & 8323 & Ford & Black & 77406 & 1 & 4742 & 3584.67 & 60488 \\
\midrule
Test
 & 5160 & BMW & Black & 61622 & 1 & 6716 & 3584.67 & 60488 \\
\bottomrule
\end{tabular}
\end{table}

Since predictions are determined by the tree structure, we extracted the branch--identical for both models—that leads to the predictions of these three observations. This branch is illustrated in Figure~\ref{fig:cart_tree}.

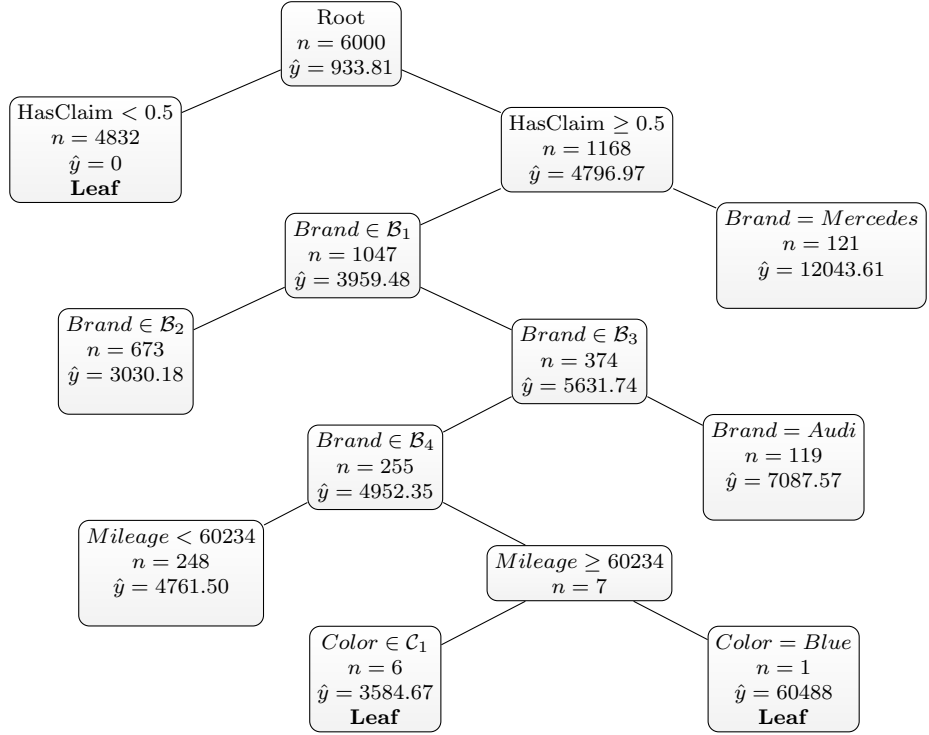
\begin{figure}[ht!]
\centering
\scriptsize
\caption{Branch of the depth-6 regression tree used for prediction}
\label{fig:cart_tree}
\begin{tikzpicture}[
  level distance=1.4cm,
  level 1/.style={sibling distance=6.5cm},
  level 2/.style={sibling distance=6.2cm},
  level 3/.style={sibling distance=6cm},
  level 4/.style={sibling distance=5.4cm},
  every node/.style={
    draw,
    rectangle,
    rounded corners,
    align=center,
    top color=white,
    bottom color=gray!10
  }
]

\node {Root\\
$n=6000$\\
$\hat{y}=933.81$}
child {
  node {HasClaim $<0.5$\\
  $n=4832$\\
  $\hat{y}=0$\\
  \textbf{Leaf}}
}
child {
  node {HasClaim $\ge 0.5$\\
  $n=1168$\\
  $\hat{y}=4796.97$}
  child {
    node {$Brand \in \mathcal{B}_1$\\
    $n=1047$\\
    $\hat{y}=3959.48$}
    child {
      node {$Brand \in \mathcal{B}_2$\\
      $n=673$\\
      $\hat{y}=3030.18$\\
      }
    }
    child {
      node {$Brand \in \mathcal{B}_3$\\
      $n=374$\\
      $\hat{y}=5631.74$}
      child {
        node {$Brand \in \mathcal{B}_4$\\
        $n=255$\\
        $\hat{y}=4952.35$}
        child {
          node {$Mileage < 60234$\\
          $n=248$\\
          $\hat{y}=4761.50$\\
          }
        }
        child {
          node {$Mileage \ge 60234$\\
          $n=7$}
          child {
            node {$Color \in \mathcal{C}_1$\\
            $n=6$\\
            $\hat{y}=3584.67$\\
            \textbf{Leaf}}
          }
          child {
            node {$Color = Blue$\\
            $n=1$\\
            $\hat{y}=60488$\\
            \textbf{Leaf}}
          }
        }
      }
      child {
        node {$Brand = Audi$\\
        $n=119$\\
        $\hat{y}=7087.57$\\
        }
      }
    }
  }
  child {
    node {$Brand = Mercedes$\\
    $n=121$\\
    $\hat{y}=12043.61$\\
    }
  }
};
\end{tikzpicture}
\scriptsize
\begin{align*}
  \mathcal{B}_1 &= \{\text{Audi, BMW, Ford, Honda, Hyundai, Kia, Nissan, Toyota, Volkswagen}\} \\
  \mathcal{B}_2 &= \{\text{Honda, Hyundai, Kia, Nissan, Toyota, Volkswagen}\}\\
  \mathcal{B}_3 &= \{\text{Audi, BMW, Ford}\}; \quad
  \mathcal{B}_4 = \{\text{BMW, Ford}\}; \quad
  \mathcal{C}_1 = \{\text{Gray, Green, Red}\}
\end{align*}


\end{figure}

From this branch of the tree with depth $6$, the difference in predictions can be explained by the fact that certain observations are not assigned to either of the two expected leaves. Specifically, the category \textit{Black} does not belong to either of the two subsets defined by the splitting conditions ($\text{Color} \in \mathcal{C}_1$ and $\text{Color} = \text{Blue}$). In the QUBO tree, the splitting rule is formulated as $\text{Color} \in \mathcal{C}_1$. When this condition is not satisfied, the observation is automatically directed to the second node. In contrast, in the RPART tree, the splitting rules are interpreted as $\text{Color} \in \mathcal{C}_1 \quad \text{or} \quad \text{Color} = \text{Blue}$. When both conditions are not satisfied, the RPART algorithm assigns the observation to the node containing the largest number of training observations. This difference in the assignment rule explains the discrepancies observed in the validation and test predictions, even though the tree structures are identical.

The discrepancy therefore does not arise from the structure of the two trees, but rather from the structure of the data, in particular the sample size. Indeed, when the sample size increases from $10{,}000$ to $50{,}000$ observations, the two trees with depth $6$ produce identical results, as illustrated in Table~\ref{tab:mse_comparison_2}.


\begin{table}[ht!]
\centering
\scriptsize
\caption{Performance comparison of regression trees (depth 6) for \texttt{datagen ($50{,}000$ observations)} dataset.}
\label{tab:mse_comparison_2}
\begin{tabular}{lccccc}
\toprule
\textbf{Method} & \textbf{Leaves} & \textbf{Depth} &
\makecell{\textbf{Train MSE} \\ \textbf{(Rel. MSE \%)}} &
\makecell{\textbf{Validation MSE} \\ \textbf{(Rel. MSE \%)}} &
\makecell{\textbf{Test MSE} \\ \textbf{(Rel. MSE \%)}} \\
\midrule
RPART & 27 & 6
  & $3.2795641\times10^{7}$
  & $4.5282717\times10^{7}$
  & $6.0606668\times10^{7}$ \\[6pt]
QUBO  & 27 & 6
  & \makecell{$3.2795641\times10^{7}$ \\ (0.00)}
  & \makecell{$4.5282717\times10^{7}$ \\ (0.00)}
  & \makecell{$6.0606668\times10^{7}$ \\ (0.00)} \\
\bottomrule
\end{tabular}
\end{table}


Finally, the comparison of maximal tree performance across the two sample sizes shows a decrease in the relative MSE on the validation and test datasets when the sample size increases, as illustrated in Table~\ref{tab:mse_max_trees}.


\begin{table}[ht!]
\centering
\scriptsize
\caption{Performance comparison of QUBO and RPART maximal trees for different sample size.}
\label{tab:mse_max_trees}
\begin{tabular}{lccccc}
\toprule
\textbf{Sample Size} & \textbf{Method} &
\makecell{\textbf{Leaves} \\ \textbf{(Depth)}} &
\makecell{\textbf{Train MSE} \\ \textbf{(Rel. MSE \%)}} &
\makecell{\textbf{Validation MSE} \\ \textbf{(Rel. MSE \%)}} &
\makecell{\textbf{Test MSE} \\ \textbf{(Rel. MSE \%)}} \\
\midrule
\multirow{2}{*}{10,000}
& RPART & \makecell{1168 \\ (26)}
  & $4.3815\times10^{-2}$
  & $3.7208\times10^{7}$
  & $1.1309\times10^{7}$ \\[6pt]
& QUBO  & \makecell{1168 \\ (26)}
  & \makecell{$4.3815\times10^{-2}$ \\ (0.00)}
  & \makecell{$3.9055\times10^{7}$ \\ (+4.96)}
  & \makecell{$1.2193\times10^{7}$ \\ (+7.82)} \\
\midrule
\multirow{2}{*}{50,000}
& RPART & \makecell{5882 \\ (30)}
  & $1.96771\times10^{5}$
  & $7.4373\times10^{7}$
  & $7.5802\times10^{7}$ \\[6pt]
& QUBO  & \makecell{5887 \\ (31)}
  & \makecell{$1.96769\times10^{5}$ \\ ($-0.001$)}
  & \makecell{$7.4393\times10^{7}$ \\ (+0.03)}
  & \makecell{$7.6084\times10^{7}$ \\ (+0.37)} \\
\bottomrule
\end{tabular}
\end{table}


These results indicate that increasing the sample size mitigates the effects related to data structure and leads to more consistent predictive performance between the two models. Furthermore, we conclude that the QUBO-based tree and the RPART tree exhibit identical structures and produce the same predictions for large-sample datasets.

\subsubsection{Discussion\\}

We begin this discussion by presenting a comparative table summarizing the performance of the optimal trees evaluated on the test data. The results are reported in Table~\ref{tab:cross_dataset_summary}.

\begin{table}[ht]
\scriptsize
\centering
\caption{Cross-dataset comparison of Test and Relative MSE (RPART Baseline).}
\label{tab:cross_dataset_summary}
\begin{tabular}{llcccc}
\toprule
\textbf{Dataset} & \textbf{Method} & \textbf{Leaves} &
\textbf{Test MSE} &
\textbf{Rel. MSE (\%)} \\
\midrule

\multirow{3}{*}{\texttt{df}}
 & QUBO  & 17 & $1.2021\times10^{6}$  & 0.00 \\
 & SK    & 21 & $1.2100\times10^{6}$  & +0.66 \\
 & RPART & 17 & $1.2021\times10^{6}$  & 0.00 \\
\midrule

\multirow{3}{*}{\texttt{datagen}}
 & QUBO  & 7 & $6.0600\times10^{7}$  & 0.00 \\
 & SK    & 47 & $6.0603\times10^{7}$  & +0.01 \\
 & RPART & 7 & $6.0600\times10^{7}$  & 0.00 \\
\midrule

\multirow{3}{*}{\texttt{freMPL}}
 & QUBO  & 16 & $2.9788\times10^{8}$  & 0.00 \\
 & SK    & 16 & $2.7434\times10^{8}$  & -7.8 \\
 & RPART & 16 & $2.9788\times10^{8}$  & 0.00 \\
\bottomrule
\end{tabular}
\end{table}

Table~\ref{tab:cross_dataset_summary} highlights the behavior of the QUBO-based decision tree across datasets of varying sizes and levels of complexity. On the simulated dataset \texttt{df}, the QUBO method exactly matches the performance of RPART and slightly outperforms the standard Python CART implementation, demonstrating its ability to recover optimal splits in a controlled setting. On the generated insurance dataset \texttt{datagen}, all three methods exhibit nearly identical predictive performance, with negligible relative differences in test MSE. This suggests that, for moderately complex and well-structured data, the QUBO formulation remains fully competitive with classical greedy approaches.
For the real-world dataset \texttt{freMPL}, the SK tree achieves a lower test error than both QUBO and RPART, with an improvement of approximately 8\%. QUBO and RPART display identical performance, indicating that the QUBO optimization strategy maintains strong generalization capabilities in the presence of noisy and heterogeneous insurance data.

Several key insights emerge from these results. First, the QUBO formulation of the split optimization problem enables solutions that are close to the exact optimum while producing clear and interpretable group separations. This opens promising avenues for applying our QUBO formulation to various fractional optimization problems, particularly in contexts where the search for globally optimal solutions is a central challenge. Second, in terms of overall predictive accuracy, QUBO-based trees deliver performance comparable to classical CART implementations without loss of interpretability. Although the predictive gain remains modest, the approach is applicable to multivariate regression and large-scale datasets using modern solvers, making it suitable for applications such as individual actuarial reserve modeling. A major advantage of this methodology is its compatibility with quantum computing. The direct mapping between QUBO problems and Ising Hamiltonians makes QUBO formulations naturally suitable for quantum hardware. As the QUBO format is natively supported by quantum annealing and gate-based solvers, our approach offers a direct path toward potentially quantum-accelerated optimal CART extensions. Overall, these results highlight the practical relevance of QUBO formulations as a robust alternative to variance-reduction based tree construction, establishing a bridge between machine learning, combinatorial optimization, quantum computing, actuarial science, and finance.

\section{Conclusion and Future Work}
The proposed QUBO formulation exactly solves the categorical split 
optimization problem in single-target regression trees, recovering 
the same optimal partitions as exhaustive search while maintaining 
the predictive accuracy of classical CART implementations across all 
datasets considered. By recasting the node-splitting criterion as a QUBO problem, we establish a connection between CART-based tree induction and combinatorial optimization. Our results show that the Gurobi solver efficiently identifies the globally optimal categorical split at each node, and that the iterative algorithm converges rapidly to the optimal solution on both simulated and real-world datasets. This confirms the feasibility and robustness of the approach across datasets with varying cardinality and structure.

Beyond exact node-level optimization, this work opens natural pathways 
toward globally optimized regression trees --- jointly optimizing splits 
across all nodes --- and toward multivariate extensions encompassing 
oblique splits and multi-target responses, where classical heuristics yields suboptimal solutions and no exact polynomial-time alternative is known. The proposed QUBO framework thus bridges statistical learning and combinatorial optimization, with potential extensions to quantum annealing solvers that could further enhance scalability for large-scale categorical predictors. Promising directions for future work include the design of hybrid classical-quantum algorithms that integrate QUBO-based node splitting into the CART framework, potentially reducing computation time while preserving predictive performance. Such developments could broaden the applicability of optimal regression trees to large-scale datasets across a wide range of domains.

\appendix

\section*{Appendix}}


\section{QUBO Formulations}
\label{Appendix: QUBO}

\subsection{Basic QUBO Formulation}

A QUBO problem consists of determining the assignment of binary variables 
$x_1,\ldots,x_n$ that minimizes a quadratic objective function. More formally, we have
\begin{equation*}
\min_{\mathbf{x} \in \{0,1\}^n} 
f(\mathbf{x}) = \sum_{i=1}^{n} a_i x_i + \sum_{i<j} b_{ij} x_i x_j,
\end{equation*}
where $x_i \in \{0,1\}$, are binary decision variables, $\mathbf{x}^\top = (x_1, x_2, \cdots x_n)$ and $a_i, b_{ij} \in \mathbb{R}$ are real coefficients. Since for binary variables $x_i \in \{0,1\}$ the identity $x_i^2 = x_i$ holds, the QUBO problem can equivalently be written in matrix form as
\begin{equation*}
\min_{\mathbf{x} \in \{0,1\}^n} f(\mathbf{x}) = \mathbf{x}^\top Q \mathbf{x},
\end{equation*}
where $Q \in \mathbb{R}^{n \times n}$ is a matrix containing both the 
linear and quadratic coefficients of the objective function. Without loss of generality, the matrix $Q$ may be assumed to be symmetric, since
\begin{equation*}
\mathbf{x}^\top Q \mathbf{x} = \mathbf{x}^\top \left(\frac{Q + Q^\top}{2}\right) \mathbf{x}.
\end{equation*}

Symmetrization can therefore be achieved by replacing each off-diagonal entry with
\begin{equation*}
q_{ij} = \frac{q_{ij} + q_{ji}}{2}.
\end{equation*}

Alternatively, a strictly upper-triangular representation may be used, 
in which each interaction term is counted only once:
\begin{equation*}
q_{ij} =
\begin{cases}
q_{ij} + q_{ji}, & \text{if } j > i, \\
0, & \text{if } j < i.
\end{cases}
\end{equation*}

\subsection{General QUBO Formulation}

In its standard form, a QUBO problem involves only binary decision variables. However, many constrained binary optimization problems can be reformulated as QUBO models using penalty methods. Consider the general binary optimization problem


\begin{align}
\min \;  f(\mathbf{x}) = \mathbf{x}^\top C \mathbf{x} \nonumber\\
\text{s.t. } 
\begin{cases}
A\mathbf{x} = b, \\
l \le B\mathbf{x} \le u,
\end{cases}
\label{ref_Q}
\end{align}


where $\mathbf{x} \in \{0,1\}^n$, $A$ and $B$ are given matrices, and $l$, $b$, and $u$ are constant vectors defining equality and inequality bounds.

Problem \eqref{ref_Q} can be transformed into a QUBO formulation by introducing quadratic penalty terms for the equality constraints. Inequalities can be converted into equalities by introducing slack variables, allowing them to be encoded in quadratic penalty terms as well \citep{glover2022quantum, tamura2021performance, bontekoe2023translating}. In particular, an equality constraint $Ax = b$ can be penalized as
\begin{equation*}
   P_{\text{eq}}(\lambda , \mathbf{x}) = \lambda \|A\mathbf{x} - b\|^2, 
\end{equation*}
where $\lambda > 0$ is a penalty parameter. By embedding constraint violations into the objective function, one obtains the penalized QUBO formulation
\begin{equation*}
   f(\mathbf{x}) + P_{\text{eq}}(\mathbf{x},\lambda) = \mathbf{x}^\top Q \mathbf{x} + c,
\end{equation*}
where $Q = C + D$ results from the expansion of the penalty term, and where $c$ is a constant independent of $\mathbf{x}$. For further details on encoding inequality constraints with slack variables in a quadratic form, see \citep{tamura2021performance, bontekoe2023translating, Xavier2026QuioQubo}.


\section{Proof of Theorems and Lemmas}
\label{app:proofLemma}
\subsection{Proof of Lemma~\ref{lemma:variance_identity}}

\begin{proof}{}
Expanding the double sum on the right-hand side of \eqref{eq:variance_identity} gives
\begin{align*}
\sum_{i=1}^{N}\sum_{j=1}^{N}(Y_i - Y_j)^2
&= N\sum_{i=1}^N Y_i^2 - 2\left(\sum_{i=1}^N Y_i\right)^2 + N\sum_{j=1}^N Y_j^2 \\
&= 2N\sum_{i=1}^N Y_i^2 - 2N^2\bar{Y}^2.
\end{align*}


Then


\begin{equation*}
    \frac{1}{2N^2}\sum_{i=1}^{N}\sum_{j=1}^{N}(Y_i-Y_j)^2
= \frac{1}{N}\sum_{i=1}^N Y_i^2 - \bar{Y}^2 = \sigma^2.
\end{equation*}
\end{proof}


\subsection{Proof of Lemma~\ref{Lemma:binary-sum-var}}
\begin{proof}{}
By defining the categories of the explanatory variable $X$ as classes and using Lemma~\ref{lemma:variance_identity}, we obtain for the subset $S_L$ associated to left child node $\mathcal{N}_{t_L}$:
\vspace{-0.5\baselineskip}
\begin{align*}
 N_L \operatorname{Var}_L & = \dfrac{1}{2 N_L} \sum_{C_{\alpha} \in S_L} \sum_{C_{\beta} \in S_L} \sum_{i \in C_{\alpha}} \sum_{j \in C_{\beta}} (Y_i - Y_j)^2 \nonumber \\
 &= \dfrac{1}{N_L} \sum_{C_{\alpha} \in S_L} \sum_{C_{\beta} \in S_L} V_{\alpha, \beta},
\end{align*}


where $C_{\alpha}, C_{\beta}$ are classes with $\alpha, \beta \in \{1, \cdots, M\}$ and the values $V_{\alpha, \beta}$ are defined by
\begin{equation*}
\label{eq_de_V}
    V_{\alpha, \beta} = \dfrac{1}{2} \sum_{i \in C_{\alpha}} \sum_{j \in C_{\beta}} (Y_i - Y_j)^2.
\end{equation*}


Then, we remark that $V_{\alpha, \beta} = V_{\beta, \alpha}$.
By introducing binary variables, the double sum over the elements of the left subset $S_L$ and right subset $S_R$ can be reformulated as a sum over all classes, so that


\begin{align*}
 N_L \operatorname{Var}_L &= \dfrac{1}{N_L} \sum_{\alpha, \beta} q_{\alpha} q_{\beta} V_{\alpha, \beta},\qquad
  N_R \operatorname{Var}_R = \dfrac{1}{N_R} \sum_{\alpha, \beta} (1- q_{\alpha})(1-q_{\beta}) V_{\alpha, \beta}.
\end{align*}


Moreover, we have


\begin{align*}
      N_L &= \sum\limits_{C_{\gamma} \in S_L} \sum\limits_{i \in C_{\gamma}} 1 = \sum\limits_{\gamma} q_{\gamma} N_{\gamma}, \qquad
      N_R = \sum\limits_{C_{\gamma} \in S_R} \sum\limits_{i \in C_{\gamma}} 1 = \sum\limits_{\gamma} (1-q_{\gamma}) N_{\gamma},
\end{align*}


where $\gamma \in \{1,\cdots,M\}$ and $N_{\gamma}$ denotes the number of observations in class $C_{\gamma}$. 
\end{proof}

\subsection{Proof of Lemma~\ref{lemma:fractional_decomposition}}
\begin{proof}{}
According to Lemma~\ref{Lemma:binary-sum-var} and recalling that $N_S$ denotes the number of observations in parent node $\mathcal{N}_t$, $\mathcal{R}(\mathbf{q})$ can be expressed as
\begingroup
\allowdisplaybreaks
\begin{align*}
  \mathcal{R}(\mathbf{q})  
  &= \dfrac{1}{N_L} \sum_{\alpha, \beta} q_{\alpha} q_{\beta} V_{\alpha, \beta} + \dfrac{1}{N_R} \sum_{\alpha, \beta} (1- q_{\alpha})(1-q_{\beta}) V_{\alpha, \beta}\\
    &= 
    \dfrac{1}{N_R N_L} 
    \left\{ 
    N_R \sum _{\alpha, \beta} q_{\alpha} q_{\beta} V_{\alpha, \beta} 
     + 
     N_L \sum _{\alpha, \beta} (1- q_{\alpha})(1-q_{\beta}) V_{\alpha, \beta}   \right\} \\ 
    &= 
    \dfrac{1}{N_R N_L} 
    \left\{ 
    \sum _{\gamma} (1-q_{\gamma}) N_{\gamma} \sum _{\alpha, \beta} q_{\alpha} q_{\beta} V_{\alpha, \beta} 
    +
    \sum _{\gamma} q_{\gamma} N_{\gamma} \sum _{\alpha, \beta} (1- q_{\alpha})(1-q_{\beta}) V_{\alpha, \beta} 
    \right\} \\
     &= 
    \dfrac{1}{N_R N_L} 
    \left\{
    \left(\sum _{\gamma} N_{\gamma}\right) \sum _{\alpha, \beta} q_{\alpha} q_{\beta} V_{\alpha, \beta} 
    + 
    \sum _{\gamma} q_{\gamma} N_{\gamma} \sum _{\alpha, \beta} (1-q_{\alpha} - q_{\beta}) V_{\alpha, \beta} 
    \right\} \\ 
    &= 
    \dfrac{1}{N_R N_L} \Bigg\{ 
    N_S \sum _{\alpha, \beta} q_{\alpha} q_{\beta} V_{\alpha, \beta} 
    + 
    \sum _{\gamma} N_{\gamma} q_{\gamma} \sum _{\alpha, \beta} V_{\alpha, \beta} \\ 
    & \qquad \qquad \quad - 
    \sum _{\gamma} \left( \sum _{\alpha, \beta} q_{\alpha} V_{\alpha, \beta}\right) N_{\gamma} q_{\gamma}  
    - 
    \sum _{\gamma} \left(\sum _{\alpha, \beta} q_{\beta} V_{\alpha, \beta} \right) N_{\gamma} q_{\gamma}  
    \Bigg\}\\ 
    &= 
     \dfrac{1}{N_R N_L} 
     \Bigg\{
     N_S \sum _{\alpha, \beta} q_{\alpha} q_{\beta} V_{\alpha, \beta}
    -
    \sum _{\alpha, \beta}q_{\alpha} q_{\beta} N_{\beta} \sum _{\gamma}  V_{\alpha, \gamma} \\
     & \qquad \qquad \quad -
     \sum _{\alpha, \beta} q_{\alpha} q_{\beta} N_{\alpha} \sum _{\gamma} V_{\gamma, \beta}
     \; +
    \sum _{\gamma} q_{\gamma} N_{\gamma} \sum _{\alpha, \beta}  V_{\alpha, \beta}  \Bigg\}\\ 
    &=
     \dfrac{1}{N_R N_L} 
     \Bigg\{  
     \sum _{\alpha, \beta} q_{\alpha} q_{\beta} 
     \left( 
     N_S V_{\alpha, \beta}
     -
    N_{\beta} \sum _{\gamma}  V_{\alpha, \gamma}
     -
    N_{\alpha} \sum _{\gamma} V_{\gamma, \beta} 
     \right) \;+\;
    \sum _{\alpha} q_{\alpha} N_{\alpha} \sum _{\gamma, \beta}  V_{\gamma, \beta}
     \Bigg\}\\
     &=
    \dfrac{1}{N_R N_L} 
    \Bigg\{  
    \sum _{\alpha, \beta} q_{\alpha} q_{\beta} 
    \left( 
    N_S V_{\alpha, \beta}
    -
    N_{\beta} \sum _{\gamma}  V_{\alpha, \gamma}
    -
    N_{\alpha} \sum _{\gamma} V_{\gamma, \beta} 
    \right) \; + \;
    \sum _{\alpha} \left( N_S ^2 \operatorname{Var}_{S}\right) N_{\alpha} q_{\alpha} 
    \Bigg\}. 
\end{align*}
\endgroup


Therefore, the expression of $\mathcal{R}(\mathbf{q})$ in terms of the precomputed values $V_{\alpha,\beta}$ under the binary encoding associated with a categorical explanatory variable is given by
\begin{align*}
    \mathcal{R}(\mathbf{q}) 
    = \frac{\mathbf{n(q)}}{\mathbf{d(q)}}, 
    \quad
    \mathbf{d(q)} = N_L N_R =
\sum\limits _{\alpha} N_S N_{\alpha} q_{\alpha}
-
\sum\limits _{\alpha, \beta} N_{\alpha} N_{\beta} q_{\alpha} q_{\beta},
\end{align*}


where $N_{\alpha}$ denotes the number of observations in class $C_{\alpha}$.
\end{proof}

\subsection{Proof of Lemma~\ref{lemma:lambda_update}}
\begin{proof}{}
Let $\mathbf{q}^*_n \in \arg\min_{\mathbf{q} \in \mathcal{Q}^*} 
F(\lambda_n, \mathbf{q})$ be the optimal solution at iteration $n$.

\textbf{Case~(i): $\mathbf{d}(\mathbf{q}^*_n) > 0$.}
By Lemma~\ref{lemma:optim_condition}, the update rule yields
\begin{equation*}
    \lambda_{n+1} 
    = \frac{\mathbf{n}(\mathbf{q}^*_n)}{\mathbf{d}(\mathbf{q}^*_n)} 
    = \mathcal{R}(\mathbf{q}^*_n),
\end{equation*}
where the last equality follows from  Lemma~\ref{lemma:fractional_decomposition}. 
Since $\mathcal{R}(\mathbf{q}) = N_L\operatorname{Var}_L + N_R\operatorname{Var}_R$ 
by definition, we obtain\\ 
$\lambda_{n+1} = N_L\operatorname{Var}_L + N_R\operatorname{Var}_R$.

\textbf{Case~(ii): $\mathbf{d}(\mathbf{q}^*_n) = 0$.}
Since $\mathbf{d}(\mathbf{q}^*_n) = N_L N_R = 0$, either $N_L = 0$ or 
$N_R = 0$, which corresponds to the trivial partition 
$\{\mathcal{N}_t, \emptyset\}$ where one child node is empty, 
i.e., $\mathbf{q}^*_n \in \{\mathbf{0}, \mathbf{1}\}$. 
In this case, the update rule 
$\lambda_{n+1} = \mathbf{n}(\mathbf{q}^*_n)/\mathbf{d}(\mathbf{q}^*_n)$ 
is undefined. Setting the size and variance of an empty node to zero gives
\begin{equation*}
    N_L\operatorname{Var}_L + N_R\operatorname{Var}_R 
    = N_S\operatorname{Var}_S + 0 
    = N_S\operatorname{Var}_S
    = \lambda_{n+1}.
\end{equation*}
This convention ensures that $\lambda_{n+1}$ remains well-defined and bounded at every iteration, consistently with 
Lemma~\ref{lemma:lambda_properties}.
\end{proof}

\subsection{Proof of Theorem~\ref{theorem:iterative_qubo}}

\begin{proof}{}
Using Lemma~(\ref{lemma:dinkelbach}), the parametric cost function associated with the fractional optimization problem is written as

\begin{equation}
\label{eq_hamil}
\begin{aligned}
   F(\lambda_n, \mathbf{q}) &=
   \sum\limits _{\alpha, \beta} 
   \left( 
    N_S V_{\alpha, \beta}
    -
    N_{\beta} \sum\limits _{\gamma}  V_{\alpha, \gamma}
    -
    N_{\alpha} \sum\limits _{\gamma} V_{\gamma, \beta} 
    \right) q_{\alpha} q_{\beta} \\
    & + 
    \sum\limits _{\alpha} \left( N_S ^2 \operatorname{Var}_{S}\right) N_{\alpha} q_{\alpha} 
    -
    \lambda _{n}
    \left(
    \sum _{\alpha} N_S N_{\alpha} q_{\alpha}
    -
    \sum _{\alpha, \beta} N_{\alpha} N_{\beta} q_{\alpha} q_{\beta}
    \right),
\end{aligned}
\end{equation}

where $\lambda _{n} =
    \dfrac{\mathbf{n(q_{n-1}^{*})}}{\mathbf{d(q_{n-1}^{*})}}$, $\lambda _0 = 0$, and
$\mathbf{q_{n-1}^{*}}$ denotes the optimal subsolution obtained at iteration $n-1$.

Rearranging the terms of (\ref{eq_hamil}) yields
\begin{align*}
  F(\lambda_n, \mathbf{q}) &=
   \sum\limits _{\alpha, \beta} 
   \left( 
    N_S V_{\alpha, \beta}
    -
    N_{\beta} \sum\limits _{\gamma}  V_{\alpha, \gamma}
    -
    N_{\alpha} \sum\limits _{\gamma} V_{\gamma, \beta} 
    +
    \lambda _{n} N_{\alpha} N_{\beta}
    \right) q_{\alpha} q_{\beta}\\
    &\qquad +
    \left( N_S ^2 \operatorname{Var}_{S} - N_S \lambda _{n} \right)
    \sum\limits _{\alpha} N_{\alpha} q_{\alpha}.
\end{align*}


Consequently,
\begin{equation}
    F(\lambda_n, \mathbf{q})
    =
    \sum _{\alpha, \beta} Q _{\alpha, \beta} q_{\alpha} q_{\beta}
    +
    \sum _{\alpha} L _{\alpha} q_{\alpha}
    = \mathbf{q}^T H\mathbf{q},
\end{equation}
where the symmetric QUBO matrix H has coefficients
\begin{align*}
\label{H_coef}
   Q _{\alpha, \beta} &=
   N_S V_{\alpha, \beta}
   -
   N_{\beta} \displaystyle\sum_{\gamma} V_{\alpha, \gamma}
   -
   N_{\alpha} \displaystyle\sum_{\gamma} V_{\gamma, \beta}
   +
   \lambda _{n} N_{\alpha} N_{\beta},\\
    L _{\alpha} &=
   \left( N_S ^2 \operatorname{Var}_{S} - N_S \lambda _{n} \right)N_{\alpha}.
\end{align*}
\end{proof}


\subsection{Proof of Lemma~\ref{F_properties}}
\begin{proof}{}
Let $N_R = \sum_\alpha (1-q_\alpha) N_\alpha = n_R(\mathbf{q}) $ and 
$N_L = \sum_\alpha q_\alpha N_\alpha = n_L(\mathbf{q})$ denote the number of observations in child nodes induced by $\mathbf{q}$.
Recall that $F(\lambda, \mathbf{q})  = \mathbf{n(q)} - \lambda\, \mathbf{d(q)}$, where
\begin{align*}
    \mathbf{n(q)} &= n_R(\mathbf{q}) \sum_{\alpha,\beta} q_\alpha q_\beta V_{\alpha,\beta} 
    + n_L(\mathbf{q}) \sum_{\alpha,\beta} (1-q_\alpha)(1-q_\beta) V_{\alpha,\beta}, \\
    \mathbf{d(q)} &= n_L(\mathbf{q})\, n_R(\mathbf{q}).
\end{align*}


We show that $F(\lambda, \mathbf{q})  = F(\lambda, \bar{\mathbf{q}})$ for all $\lambda \geq 0$ 
and $q \in \{0,1\}^M$, and consider two cases.

\textbf{Case 1: $\mathbf{d(q)} > 0$.} Since $\bar{\mathbf{q}} = \mathbf{1} - \mathbf{q}$, we have
\[
n_L(\bar{\mathbf{q}}) = \sum_\alpha (1-q_\alpha) N_\alpha = N_S - \sum_\alpha q_\alpha 
N_\alpha = n_R(\mathbf{q}),
\]
and similarly $n_R(\bar{\mathbf{q}}) = n_L(\mathbf{q})$, so that 
\[\mathbf{d}(\bar{\mathbf{q}}) = 
n_L(\bar{\mathbf{q}})\,n_R(\bar{\mathbf{q}}) = n_R(\mathbf{q})\,n_L(\mathbf{q}) = \mathbf{d(q)}.
\]
Substituting into the numerator yields
\begin{align*}
    \mathbf{n}(\bar{\mathbf{q}}) 
    &= n_R(\bar{\mathbf{q}}) \sum_{\alpha,\beta} (1-q_\alpha)(1-q_\beta) V_{\alpha,\beta}
    + n_L(\bar{\mathbf{q}}) \sum_{\alpha,\beta} q_\alpha q_\beta V_{\alpha,\beta} \\
    &= n_L(\mathbf{q}) \sum_{\alpha,\beta} (1-q_\alpha)(1-q_\beta) V_{\alpha,\beta}
    + n_R(\mathbf{q}) \sum_{\alpha,\beta} q_\alpha q_\beta V_{\alpha,\beta}\\ 
    &= \mathbf{n(q)}.
\end{align*}


Hence $F(\lambda, \bar{\mathbf{q}}) = \mathbf{n(q)} - \lambda\,\mathbf{d(q)} = F(\lambda, \mathbf{q}) $.

\textbf{Case 2: $\mathbf{d(q)} = 0$.} This holds if and only if $n_L(\mathbf{q}) = 0$ or $n_R(\mathbf{q}) = 0$, which corresponds to $\mathbf{q} \in \{\mathbf{0}, \mathbf{1}\}$. 
In both cases, all cross terms in $\mathbf{n(q)}$ vanish, so $\mathbf{n(q)} = 
\mathbf{d(q)} = 0$, and by the same argument $\mathbf{n}(\bar{\mathbf{q}}) = 
\mathbf{d}(\bar{\mathbf{q}}) = 0$. Therefore,
$F(\lambda, \mathbf{q})  = F(\lambda, \bar{\mathbf{q}}) = 0$,
which in particular gives $F(\lambda, \mathbf{0}) = F(\lambda, \mathbf{1}) = 0$ 
for any $\lambda \geq 0$, establishing both properties of the theorem.
\end{proof}

\subsection{Proof of Lemma~\ref{lemma:convergence_equivalence}}
\begin{proof}{}
Recall that $F(\lambda, \mathbf{q}) = \mathbf{n}(\mathbf{q}) 
- \lambda\,\mathbf{d}(\mathbf{q})$, and that $\lambda_n = 
\mathbf{n}(\mathbf{q}^*_{n-1})/\mathbf{d}(\mathbf{q}^*_{n-1})$ 
is the updated parameter at iteration $n$.

\textit{($\Leftarrow$)}
Suppose first that $\mathbf{q}^*_n \in \{\mathbf{0}, \mathbf{1}\}$.
Then $F(\lambda_n, \mathbf{q}^*_n) = 0$ for any $\lambda_n \geq 0$ 
by Lemma~\ref{F_properties}.

Suppose now that $(\lambda_n, \mathbf{q}^*_n) = (\lambda^*, \mathbf{q}^*)$.
Then $F(\lambda^*, \mathbf{q}^*) = \mathbf{n}(\mathbf{q}^*) - 
\lambda^*\,\mathbf{d}(\mathbf{q}^*) = 0$
by definition of $\lambda^* = \mathbf{n}(\mathbf{q}^*)/\mathbf{d}(\mathbf{q}^*)$
and Lemma~\ref{lemma:optim_condition}.

\textit{($\Rightarrow$)}
Suppose $F(\lambda_n, \mathbf{q}^*_n) = 0$, i.e., 
$\mathbf{n}(\mathbf{q}^*_n) = \lambda_n\,\mathbf{d}(\mathbf{q}^*_n)$.

\textbf{Case 1: $\mathbf{d}(\mathbf{q}^*_n) = 0$.}
Then $\mathbf{n}(\mathbf{q}^*_n) = 0$, which implies 
$\mathbf{q}^*_n \in \{\mathbf{0}, \mathbf{1}\}$ 
by the proof of Lemma~\ref{F_properties}.

\textbf{Case 2: $\mathbf{d}(\mathbf{q}^*_n) > 0$.}
Then $\mathbf{q}^*_n \in \mathcal{Q}^*$ is a non-trivial solution, and by definition of $\lambda^* = \min_{\mathbf{q} \in \mathcal{Q}^*} 
\mathcal{R}(\mathbf{q})$, we have
$
\lambda_n 
= \mathbf{n}(\mathbf{q}^*_n)/\mathbf{d}(\mathbf{q}^*_n) 
= \mathcal{R}(\mathbf{q}^*_n) 
\geq \lambda^*$.

It remains to show that $\lambda_n \leq \lambda^*$, and hence 
$\mathbf{q}^*_n = \mathbf{q}^*$.
Since $\mathbf{q}^*_n$ minimizes $F(\lambda_n, \cdot)$ over 
$\{0,1\}^M$ and $F(\lambda_n, \mathbf{q}^*_n) = 0$, we have 
$F(\lambda_n, \mathbf{q}') \geq 0$ for all $\mathbf{q}' \in \{0,1\}^M$.
In particular, evaluating at the global optimum $\mathbf{q}^*$:
\[
0 \leq F(\lambda_n, \mathbf{q}^*) 
= \mathbf{n}(\mathbf{q}^*) - \lambda_n\,\mathbf{d}(\mathbf{q}^*)
= \mathbf{d}(\mathbf{q}^*)\bigl(\mathcal{R}(\mathbf{q}^*) - \lambda_n\bigr)
= \mathbf{d}(\mathbf{q}^*)(\lambda^* - \lambda_n).
\]
Since $\mathbf{d}(\mathbf{q}^*) > 0$, this yields $\lambda^* \geq \lambda_n$.
Combined with $\lambda_n \geq \lambda^*$, we conclude $\lambda_n = \lambda^*$,
and therefore $\mathcal{R}(\mathbf{q}^*_n) = \mathcal{R}(\mathbf{q}^*) = \lambda^*$,
which means $\mathbf{q}^*_n = \mathbf{q}^*$.
\end{proof}

\subsection{Proof of Theorem~\ref{theorem:convergence}}
\begin{proof}
\textit{($\Rightarrow$)} Suppose there exists $k \geq 1$ such that 
$F(\lambda_k, \mathbf{q}^*_k) = 0$ with $\mathbf{q}^*_k \in \mathcal{Q}^*$. By 
Lemma~\ref{lemma:convergence_equivalence}, this implies $(\lambda_k, \mathbf{q^*_k}) = 
(\lambda^*, \mathbf{q^*})$. The update rule of Lemma~\ref{lemma:lambda_update} then gives
\[
\lambda_{k+1} = \mathbf{n(q^*_k)}/\mathbf{d(q^*_k)} = \mathcal{R}(\mathbf{q^*_k}) 
= \lambda^* = \lambda_k,
\]
so the sequence has stabilized and the algorithm returns the optimal solution.
\end{proof}

\subsection{Proof of Lemma~\ref{lemma:lambda_properties}}

\begin{proof}
\emph{Property 1: $\lambda^* \leq N_S\operatorname{Var}_S$.}

By Lemma~\ref{lemma:optim_condition}, the optimal parameter satisfies 
$\lambda^* = \mathcal{R}\mathbf{(q^*)} = N_L\operatorname{Var}_L + 
N_R\operatorname{Var}_R.$\\
Since $\mathbf{q^*}$ is a non-trivial solution, both child 
nodes are non-empty, and the weighted variance decomposition gives
\[
N_L\operatorname{Var}_L + N_R\operatorname{Var}_R \leq 
N_S\operatorname{Var}_S,
\]
as splitting a node cannot increase the total weighted variance. Therefore 
$\lambda^* \leq N_S\operatorname{Var}_S$.

\textit{Property 2: $F(\lambda, \mathbf{q})  > 0$ for any $\lambda < \lambda^*$ and 
any non-trivial $\mathbf{q} \in \mathcal{Q^*}$.}

For any non-trivial $\mathbf{q}$, we have $\mathbf{d(q)} > 0$, so
\[
F(\lambda, \mathbf{q})  = \mathbf{d(q)}\left(\frac{\mathbf{n(q)}}{\mathbf{d(q)}} - 
\lambda\right) = \mathbf{d(q)}\left(\mathcal{R}(\mathbf{q}) - \lambda\right).
\]
Since $\lambda < \lambda^* = \min\limits_{\mathbf{q} \in \mathcal{Q^*}} 
\mathcal{R}(\mathbf{q})$, we have $\mathcal{R}(\mathbf{q}) \geq \lambda^* > \lambda$ for all non-trivial $\mathbf{q}$, hence $F(\lambda, \mathbf{q})  > 0$. Furthermore, by 
Lemma~\ref{F_properties}, $F(\lambda, \mathbf{0}) = F(\lambda, \mathbf{1}) 
= 0 < F(\lambda, \mathbf{q}) $ for any non-trivial $\mathbf{q}$, so the minimum is attained by 
a trivial solution.

\textit{Property 3: For any $\lambda \geq N_S\operatorname{Var}_S$, there exists a non-trivial $\mathbf{q}$ such that $F(\lambda, \mathbf{q})  \leq 0$, and the minimum is attained by a non-trivial solution.}

Let $\mathbf{q^*}$ denote the optimal non-trivial solution. By Property 1, $\lambda^* \leq N_S\operatorname{Var}_S \leq \lambda$, so
\[
F(\lambda, \mathbf{q^*}) = \mathbf{d}(\mathbf{q^*})\left(\mathcal{R}(\mathbf{q^*}) - \lambda\right) 
= \mathbf{d}(\mathbf{q^*})(\lambda^* - \lambda) \leq 0.
\]
If $\lambda > \lambda^*$, then $F(\lambda, \mathbf{q^*}) < 0$, and since $F(\lambda, \mathbf{0}) = F(\lambda, \mathbf{1}) = 0 > F(\lambda, \mathbf{q^*})$, 
the minimum is attained by a non-trivial solution. If $\lambda = \lambda^*$, 
then $F(\lambda^*, \mathbf{q^*}) = 0$ and the minimum is also attained at $\mathbf{q^*}$ by Lemma~\ref{lemma:optim_condition}.

\textit{Property 4: Step 1 --- $(\lambda_n)$ is decreasing and bounded below by $\lambda^*$.}

We proceed by induction on $n$, showing that $\lambda^* \leq \lambda_{n+1} \leq \lambda_n$ at each step.

\textit{Base case ($n = 1$):} By assumption, $\lambda_0 \geq 
N_S\operatorname{Var}_S$. By Property 3, the minimizer $\mathbf{q_0^*}$ of 
$F(\lambda_0, \cdot)$ is non-trivial, so the update rule of 
Lemma~\ref{lemma:lambda_update} gives
$\lambda_1 = \mathcal{R}(\mathbf{q_0^*}) \leq N_S\operatorname{Var}_S \leq \lambda_0$,  
where the first inequality follows from Property 1. Furthermore, since $\mathbf{q_0^*}$ is a non-trivial solution, $\lambda_1 = \mathcal{R}(\mathbf{q_0^*}) \geq \lambda^*$ by definition of $\lambda^*$ as the minimum of $\mathcal{R}(\mathbf{q})$ over all non-trivial solutions. Hence $\lambda^* \leq \lambda_1 \leq \lambda_0$.

If $\lambda_1 = \lambda^*$, the algorithm stops and the result holds. Otherwise, we proceed to the inductive step.

\textit{Inductive step:} Assume that $\lambda^* \leq \lambda_k \leq \lambda_{k-1}$ for all $k \leq n$, so that $\lambda^* \leq \lambda_n \leq \cdots \leq \lambda_0$. We show that $\lambda^* \leq \lambda_{n+1} \leq \lambda_n$.

Let $\mathbf{q_n^*}$ be any minimizer of $F(\lambda_n, \cdot)$ over $\{0,1\}^M$. Since $F(\lambda_n, \mathbf{0}) = 0$ by Lemma~\ref{F_properties}, we have
$F(\lambda_n, \mathbf{q_n^*}) \leq F(\lambda_n, \mathbf{0}) = 0$.

\textbf{Case 1: $\mathbf{q_n^*}$ is non-trivial.} Then $\mathbf{d}(\mathbf{q_n^*}) > 0$ and
$\mathbf{d}(\mathbf{q_n^*})\bigl(\mathcal{R}(\mathbf{q_n^*}) - \lambda_n\bigr) = 
F(\lambda_n, \mathbf{q_n^*}) \leq 0$.
So $\lambda_{n+1} = \mathcal{R}(\mathbf{q_n^*}) \leq \lambda_n$.\\ Furthermore, since $q_n^*$ is non-trivial, $\lambda_{n+1} = \mathcal{R}(\mathbf{q_n^*}) \geq \lambda^*$ by definition of $\lambda^*$.

If $F(\lambda_n, \mathbf{q_n^*})=0$, then $\lambda_{n+1}=\lambda_n$ and by Theorem~\ref{theorem:convergence}, algorithm converges with $\lambda_{n+1}=\lambda_n = \lambda ^*$.

\textbf{Case 2: $\mathbf{q_n^*}$ is trivial, i.e., $\mathbf{q_n^*} \in \{\mathbf{0}, \mathbf{1}\}$}. Then $F(\lambda_n, \mathbf{q_n^*}) = 0$ by Lemma~\ref{F_properties}, and since $\mathbf{q_n^*}$ minimizes $F(\lambda_n, \cdot)$, every $\mathbf{q}$ satisfies $F(\lambda_n, \mathbf{q}) \geq 0$. In particular, for the optimal non-trivial solution $\mathbf{q_n^*}$, we have
$0 \leq F(\lambda_n, \mathbf{q_n^*}) = \mathbf{d(q^*)}(\lambda^* - \lambda_n) \leq 0$, 
where the last inequality uses $\lambda^* \leq \lambda_n$ from the inductive hypothesis. Hence $\lambda_n = \lambda^*$, and the update rule gives 
$\lambda_{n+1} = \mathcal{R}(\mathbf{q^*}) = \lambda^* = \lambda_n$, so the sequence 
has stabilized at $\lambda^*$.

\textit{Step 2 --- $(\lambda_n)$ converges to $\lambda^*$.}

Since $(\lambda_n)$ is monotonically decreasing and bounded below by $\lambda^* \geq 0$ (Property 1), it converges to some limit $\bar{\lambda} \geq \lambda^*$.

By the stopping criterion of Theorem~\ref{theorem:convergence}, the algorithm terminates at the first iteration $k$ such that $\lambda_{k+1} = \lambda_k$, 
i.e., when $\lambda_{k+1} = \mathcal{R}(\mathbf{q_k^*}) = \lambda_k$.
By Lemma~\ref{lemma:optim_condition}, this is equivalent to $F(\lambda_k, \mathbf{q_k^*}) = 0$ with $\mathbf{q_k^*} \in \mathcal{Q^*}$, which by Lemma~\ref{lemma:convergence_equivalence} implies $(\lambda_k, \mathbf{q_k^*}) = (\lambda^*, \mathbf{q_k^*})$. Hence $\bar{\lambda} = \lambda^*$.
\end{proof}

\section{Building and Pruning Tree Algorithms} 
\label{Appendix:algo build tree}
We introduce the following notation: $R(t)$ denotes the empirical risk at an internal node $\mathcal{N}_t$, weighted by the proportion of observations reaching this node; $\mathcal{A}^t$ denotes the subtree of $\mathcal{A}$ rooted at $\mathcal{N}_t$; and $R(\mathcal{A}^t)$ denotes the total empirical risk of the subtree $\mathcal{A}^t$. The tree construction algorithm and explicit computation of the pruning sequence is summarized in Algorithms~ \ref{alg:cart_qubo}, \ref{alg:cost_complexity} (see Appendix~\ref{Appendix:algo build tree}).

\begin{algorithm}[H]
\small
\caption{CART--QUBO Tree Construction}
\label{alg:cart_qubo}
\begin{algorithmic}[1]
\Require Dataset $D$, control parameters $(cp, maxdepth, minsplit)$
\Ensure Binary regression or classification tree $T$

\Function{BuildTree}{$D$}
    \If{stopping criteria are met (node size, depth, or complexity)}
        \State \Return leaf node with prediction (mean or majority class)
    \EndIf

    \ForAll{variables $X_j$ in $D$}
        \If{$X_j$ is continuous or ordinal}
            \State Compute optimal split $\delta _j$ by threshold scanning (standard CART)
        \ElsIf{$X_j$ is categorical}
            \State Compute optimal split $\delta _j$ by solving the associated QUBO problem
        \EndIf
    \EndFor

    \State Select the optimal split $\delta^* = \arg\min \limits _{\delta _j}\{\text{Impurity}\}$
    \State Partition $D$ into $(D_{\text{left}}, D_{\text{right}})$ according to $\delta^*$

    \State Node.left $\gets$ \Call{BuildTree}{$D_{\text{left}}$}
    \State Node.right $\gets$ \Call{BuildTree}{$D_{\text{right}}$}

    \State \Return Node
\EndFunction
\end{algorithmic}
\end{algorithm}


\begin{algorithm}[H]
\small
\caption{Cost--complexity pruning}
\label{alg:cost_complexity}
\begin{algorithmic}[1]
\Require Maximal tree $\mathcal{A}_0$
\Ensure Sequences $(\alpha_m)_m$ and $(\mathcal{A}_{\alpha_m})_m$

\State Initialize $\alpha_0 \gets 0$, $\mathcal{A}_{\alpha_0} \gets \mathcal{A}_0$, $m \gets 0$
\While{$\mathcal{A}_{\alpha_m}$ is not reduced to the root node}
    \ForAll{internal nodes $t$ of $\mathcal{A}_{\alpha_m}$}
        \State $g(t) \gets \dfrac{R(t) - R(\mathcal{A}_{\alpha_m}^t)}{|\mathcal{A}_{\alpha_m}^t| - 1}$
    \EndFor
    \State $t_m \gets \arg\min_t g(t)$
    \State $\alpha_{m+1} \gets g(t_m)$
    \State $\mathcal{A}_{\alpha_{m+1}} \gets \mathcal{A}_{\alpha_m} \setminus \mathcal{A}_{\alpha_m}^{t_m}$
    \State $m \gets m + 1$
\EndWhile
\State \Return $(\alpha_m)_m$, $(\mathcal{A}_{\alpha_m})_m$
\end{algorithmic}
\end{algorithm}


\section{Datasets}
\label{Appendix: datasets}
\subsection{Synthetic Insurance Data Generation Procedure}
\begin{algorithm}[H]
\small
\caption{Structured Synthetic Dataset Generation: \texttt{df}}
\label{alg:df_generation}
\begin{algorithmic}[1]

\Require Number of observations $n = 20{,}000$; random seed $123$
\Ensure Dataset $\mathcal{D}$
\State Define brand base prices $Base_{Brand}$
\State Define color multiplicative factors $c_{Color}$

\For{$i = 1$ to $n$}

    \State Draw $Brand_i$ uniformly from the set of brands
    \State $Base_i \leftarrow Base_{Brand_i}$
    
    \State Generate mileage:
    \vspace{-0.5em}
    \[
    Mileage_i \sim \text{Gamma}(shape=2.0,\ scale=30000)
    \]
    \State Truncate $Mileage_i \leftarrow \min(Mileage_i, 250{,}000)$
    
    \State Draw $Color_i$ uniformly from the set of colors
    \State $c_i \leftarrow c_{Color_i}$
    
    \State Compute claim probability (logistic form):
    \[
    p_i = \left[1 + \exp\left(-\frac{Mileage_i - 80000}{20000}\right)\right]^{-1}
    \]
    \State Draw $HasClaim_i \sim \text{Bernoulli}(p_i)$
    
    \If{$HasClaim_i = 1$}
    
        \State Compute structured severity:
        \begin{align*}
          Severity_i &=
        0.15\,Base_i
        + 0.002\,Mileage_i
        + 5000\,\mathbb{I}_{Brand_i \in \{\text{BMW, Audi, Mercedes}\}}\\
        &\qquad + 3000\,\mathbb{I}_{Color_i = \text{Red}}
        \end{align*}
        
        \State Generate noise:
        \vspace{-0.5em}
        \[
        \varepsilon_i \sim \mathcal{N}(0, 2000)
        \]
        
        \State Compute claim amount:
        \[
        ClaimAmount_i =
        \max\big(100,\; Severity_i \cdot c_i + \varepsilon_i \big)
        \]
        
    \Else
    
        \State $ClaimAmount_i \leftarrow 0$
        
    \EndIf
    
    \State Add $(Brand_i, Color_i, Mileage_i, HasClaim_i, ClaimAmount_i)$ to $\mathcal{D}$

\EndFor

\State \Return $\mathcal{D}$

\end{algorithmic}
\end{algorithm}

\begin{algorithm}[H]
\small
\caption{Synthetic Dataset Generation: \texttt{datagen}}
\label{alg:datagen_generation}
\begin{algorithmic}[1]

\Require Number of observations $n \in \{10{,}000; 50{,}000\}$, random seed $123$
\Ensure Dataset $\mathcal{D}$

\State Define brand-specific parameters $(Base_{b}, \sigma_{b})$
\State Define color multiplicative factors $c_{color}$

\For{$i = 1$ to $n$}

    \State Draw $Brand_i$ uniformly from the set of brands
    \State $Base_i \leftarrow Base_{Brand_i}$, $\sigma_i \leftarrow \sigma_{Brand_i}$
    
    \State Generate mileage:
    \vspace{-0.5em}
    \[
    Mileage_i \sim \text{LogNormal}(10,\,0.5)
    \]
    \State Truncate $Mileage_i \leftarrow \min(Mileage_i, 300{,}000)$
    
    \State Draw $Color_i$ uniformly from the set of colors
    \State $c_i \leftarrow c_{Color_i}$
    
    \State Compute claim probability:
    \[
    p_i = \min\big(\max(0.15 + 0.000002 \times Mileage_i,\, 0.01),\, 0.9\big)
    \]
    \State Draw $HasClaim_i \sim \text{Bernoulli}(p_i)$
    
    \If{$HasClaim_i = 1$}
    
        \State Generate base severity:
        \[
        Basic_i \sim 
        \text{LogNormal}\big(\log(0.1 \times Base_i),\, \sigma_i\big)
        \]
        
        \State Compute mileage component:
        \vspace{-0.5em}
        \[
        KM_i = 0.001 \times Mileage_i \times U(0.5,1.5)
        \]
        
        \State With probability $0.02$, generate heavy tail:
        \[
        Tail_i \sim \text{LogNormal}(\log(Base_i),\,1.0)
        \]
        \State Otherwise $Tail_i \leftarrow 0$
        
        \State Compute claim amount:
        \[
        ClaimAmount_i =
        \max\big(50,\; (Basic_i + KM_i + Tail_i)\times c_i \big)
        \]
        
    \Else
    
        \State $ClaimAmount_i \leftarrow 0$
        
    \EndIf
    
    \State Add $(Brand_i, Color_i, Mileage_i, HasClaim_i, ClaimAmount_i)$ to $\mathcal{D}$

\EndFor
\State \Return $\mathcal{D}$
\end{algorithmic}
\end{algorithm}

\newpage
\subsection{Real-World Datasets variables}


\begin{table}[ht]
\centering
\small
\caption{Variables used in the experiments for each dataset.}
\label{tab:datasets_variables}
\begin{tabular}{l p{10cm}}
\toprule
\textbf{Dataset} & \textbf{Variables} \\
\midrule
\textbf{Ames}
  & \textit{SalePrice$^*$, MSZoning, HouseStyle, BldgType} \\
\midrule
\textbf{freMPL}
  & \textit{Exposure, VehBody, VehClass, DrivAge, VehEnergy, RiskVar, ClaimAmount, Y$^*$} \\
\midrule
\textbf{ausprivauto}
  & \textit{Exposure, VehValue, VehAge, VehBody, DrivAge, ClaimAmount, Y$^*$} \\
\bottomrule
\multicolumn{2}{l}{%
  \footnotesize $^*$Response variable,  $Y = \textit{ClaimAmount}/\textit{Exposure}$%
}\\
\end{tabular}
\end{table}


\section{Optimization results (Convergence Behavior)}
\label{Appendix: Optim results}

\begin{table}[H]
\centering
\scriptsize
\caption{Convergence of the QUBO algorithm for the \texttt{df\_tricky\_means} dataset}
\label{tab:tricky_color_iterations}
\begin{tabular}{c c c c c c}
\toprule
Variable & Iteration & $\lambda_{\text{initial}}$ & Binary partition vector & Score & $\lambda_{\text{final}}$ \\ 
\midrule

\multirow{3}{*}{Color}
& 1
& $5.8136\times10^{10}$
& $(0,0,0,1,1,1)$
& $1.1516 \times10^{7}$
& $5.7594\times10^{10}$ \\

& 2
& $5.7594\times10^{10}$
& $(1,1,1,0,1,0)$
& $1.1514 \times10^{7}$
& $5.7583\times10^{10}$ \\

& 3
& $5.7583\times10^{10}$
& $\mathbf{(0,0,0,1,0,1)}$
& $1.1514 \times10^{7}$
& $5.7583\times10^{10}$ \\
\bottomrule
\end{tabular}


\footnotesize
\scriptsize
\textbf{Binary vector encoding:}

\textit{Color:} $(q_1,\dots,q_6) = (\text{Black, Blue, Gray, Green, Red, White})$.
\end{table}


\begin{table}[H]
\centering
\scriptsize
\caption{Convergence of the QUBO algorithm on the \texttt{Ames} dataset.}
\label{tab:ames_convergence}
\begin{tabular}{l l l l l l}
\toprule
Variable & Iteration & $\lambda_{\text{initial}}$ & Binary vector & Score & $\lambda_{\text{final}}$ \\ 
\midrule

\multirow{3}{*}{HouseStyle}
& 1 & 0 & $(0,0,0,0,0,0,0,0)$ & 0.631 & $9.2079 \times 10^{2}$ \\
& 2 & $9.2079 \times 10^{2}$ & $(0,0,0,1,0,1,0,0)$ & 0.592 & $8.6438 \times 10^{2}$ \\
& 3 & $8.6438 \times 10^{2}$ & $\mathbf{(1,1,1,0,1,0,1,1)}$ & 0.592 & $8.6438\times 10^{2}$ \\ 
\midrule

\multirow{3}{*}{BldgType}
& 1 & 0 & $(0,0,0,0,0)$ & 0.631 & $9.2079 \times 10^{2}$ \\
& 2 & $9.2079 \times 10^{2}$ & $(0,1,1,1,0)$ & 0.609 & $8.8925 \times 10^{2}$ \\
& 3 & 889.246 & $\mathbf{(1,0,0,0,1)}$ & 0.609 & $8.8925 \times 10^{2}$ \\ 
\bottomrule
\end{tabular}

\footnotesize
\scriptsize
\textbf{Binary vector encoding:}

\textit{HouseStyle:} $(q_1,\dots,q_8)=(\text{1.5Fin, 1.5Unf, 1Story, 2.5Fin, 2.5Unf, 2Story, SFoyer, SLvl})$.

\textit{BldgType:} $(q_1,\dots,q_5)=(\text{1Fam, 2fmCon, Duplex, Twnhs, TwnhsE})$.
\end{table}


\begin{table}[H]
\centering
\scriptsize
\renewcommand{\arraystretch}{1.2}
\caption{Convergence of the QUBO algorithm (Exact and Gurobi Solveur) on the \textit{ausprivauto} dataset.}
\label{tab:ausprivauto_conv}
\begin{tabular}{l l l l l l}
\toprule
Variable & Iteration & $\lambda_{\text{initial}}$ & Binary vector & Score & $\lambda_{\text{final}}$ \\ 
\midrule

\multirow{3}{*}{\makecell[l]{VehBody \\ Exact}}
&  1 & 0  & $(0,0,0,0,0,0,0,0,0,0,0,0,0)$ & $1.8438 \times 10 ^{10}$ & 
$6.3943 \times 10 ^{13}$ \\ 

& 2 & $6.3943 \times 10 ^{13}$ & $(0,0,0,0,0,0,0,0,0,1,0,1,0)$ & $1.8406 \times 10 ^{10}$ & $ 6.3833 \times 10 ^{13}$ \\ 

& 3 & $ 6.3833 \times 10 ^{13}$ & $\mathbf{(1,1,1,1,1,1,1,1,1,0,1,0,1)}$ & $1.8406 \times 10 ^{10}$ & $ 6.3833 \times 10 ^{13}$ \\ 
\midrule

\multirow{3}{*}{\makecell[l]{VehBody \\ Gurobi}}
&  1 & 0  & $(0,0,0,0,0,0,0,0,1,0,0,0,0)$ & $1.8438 \times 10 ^{10}$& $6.3942\times 10 ^{13}$ \\

& 2 & $6.3942\times 10 ^{13}$ & $(1,1,1,1,1,1,1,1,1,0,1,0,1)$ & $1.8406 \times 10 ^{10}$ & $6.3833\times 10 ^{13}$ \\

& 3 & $6.3833\times 10 ^{13}$ & $\mathbf{(0,0,0,0,0,0,0,0,0,1,0,1,0)}$ & $1.8406 \times 10 ^{10}$ & $6.3833\times 10 ^{13}$\\ 
\midrule

\multirow{3}{*}{\makecell[l]{VehClass \\ Exact}}
& 1 & 0 & $(0,0,0,0,0,0)$ & $1.8438\times 10 ^{10}$ & $6.3943\times 10 ^{13}$ \\

& 2 & $6.3943\times 10 ^{13}$ & $(0,0,0,0,0,1)$ & $1.8411 \times 10 ^{10}$ & $6.3848\times 10 ^{13}$\\

& 3 & $6.3848\times 10 ^{13}$ & $\mathbf{(0,0,0,0,0,1)}$ & $1.8411 \times 10 ^{10}$ & $6.3848\times 10 ^{13}$ \\
 \midrule
 
 \multirow{3}{*}{\makecell[l]{VehClass \\ Gurobi}}
& 1 & 0 & $(1,1,0,1,1,1)$ & $1.8438 \times 10 ^{10}$ & $6.39423\times 10 ^{13}$ \\

& 2 & $6.3942\times 10 ^{13}$ & $(1,1,1,1,1,0)$ & $1.8411 \times 10 ^{10}$ & $6.3848\times 10 ^{13}$\\

& 3 & $6.3848\times 10 ^{13}$ & $\mathbf{(1,1,1,1,1,0)}$ & $1.8411 \times 10 ^{10}$ & $6.3848\times 10 ^{13}$ \\
 \bottomrule
\end{tabular}

\footnotesize
\scriptsize
\textbf{Binary vector encoding:}  

\textit{VehBody :} $(q_1, q_2, \cdots, q_{13}) = $ (Bus, Convertible, Coupe, Hardtop, Hatchback, Minibus, Motorized caravan, Panel van, Roadster, Sedan, Station wagon, Truck, Utility). 

\textit{VehClass :}  $(q_1, q_2, \cdots, q_{6}) =$ (old people, older work. people, oldest people, working people, young people, youngest people).
\end{table}

\bibliographystyle{plainnat}
\bibliography{ref}

\end{document}